\documentclass[preprint,12pt]{elsarticle}

\usepackage[utf8]{inputenc}
\usepackage[T1]{fontenc}
\usepackage{amsmath}
\usepackage{amssymb}
\usepackage{amsthm}
\usepackage{xspace}
\usepackage{xfrac}
\usepackage[colorlinks = true]{hyperref}
\usepackage{lmodern}

\usepackage[shortlabels]{enumitem}

\setlist[description]{align = right,%
					  labelwidth = 2cm,%
					  font = \itshape\mdseries,%
					  topsep = .6em,%
					  parsep = 0pt,%
					  itemsep = .5em,%
					  listparindent = 1.2cm}
\setlist[enumerate]{(i),%
					topsep = .5em,%
					parsep = 0pt}				

\usepackage
[
    subrefformat = simple,
]
{subcaption}

\usepackage
[
    noabbrev,
    nameinlink,
]
{cleveref}

\crefname{stat}{statement}{statements}
\creflabelformat{stat}{#2{#1}#3}

\crefname{lem}{lemma}{lemmas}
\creflabelformat{lem}{#2{#1}#3}

\newcommand*{\emDash}{\hspace*{.25pt}--\hspace*{.25pt}}

\newtheorem{thm}{Theorem}
\newtheorem{lem}[thm]{Lemma}
\newtheorem{obs}[thm]{Observation}

\DeclareMathOperator{\Or}{O}

\DeclareMathOperator{\Tr}{Tr}

\newcommand*{\nwspace}{\hspace*{.1em}}
\newcommand*{\fd}{\,{\to}\,}

\newcommand*{\Dyp}{\mathcal{D}}
\newcommand*{\Fyp}{\mathcal{F}}
\newcommand*{\Gyp}{\mathcal{G}}
\newcommand*{\Hyp}{\mathcal{H}}

\newcommand*{\var}[1]{\texttt{Var}_{#1}}
\newcommand*{\N}{\mathbb{N}}

\newcommand*{\rel}{\mathfrak{r}}
\newcommand*{\sel}{\mathfrak{s}}

\newcommand{\true}{\textsc{true}\xspace}  
\newcommand{\false}{\textsc{false}\xspace} 
\newcommand*{\FPT}{\textsf{FPT}\xspace}
\newcommand*{\co}{\textsf{co}}
\renewcommand*{\P}{\textsf{P}\xspace}
\newcommand*{\NP}{\textsf{NP}\xspace}
\newcommand*{\W}{\textsf{W}\xspace}
\newcommand*{\poly}{\textsf{poly}\xspace}

\newcommand{\FD}{\textsc{Functional Dependency}\xspace}
\newcommand{\EnumFD}{\textsc{Enumerate Minimal FDs}\xspace}
\newcommand{\FDfixed}{\textsc{Functional Dependency}\textsubscript{\textit{fixed RHS}}\xspace}
\newcommand{\EnumFDfixed}{\textsc{Enumerate Minimal FDs}\textsubscript{\textit{fixed RHS}}\xspace}
\newcommand{\HS}{\textsc{Hitting Set}\xspace}
\newcommand{\IND}{\textsc{Inclusion Dependency}\xspace}
\newcommand{\EnumIND}{\textsc{Enumerate Maximal INDs}\xspace}
\newcommand{\INDfixed}{\textsc{Inclusion Dependency}\textsubscript{\textit{Identity}}\xspace}
\newcommand{\EnumINDfixed}{\textsc{Enumerate Maximal INDs}\textsubscript{\textit{Identity}}\xspace}
\newcommand{\TransHyp}{\textsc{Transversal Hypergraph}\xspace}
\newcommand{\TransHypUnion}{\textsc{Transversal Hypergraph Union}\xspace}
\newcommand{\Unique}{\textsc{Unique Column Combination}\xspace}
\newcommand{\EnumUnique}{\textsc{Enumerate Minimal UCCs}\xspace}
\newcommand*{\WANS}{\textsc{Weighted Antimonotone $3$-normalized Satisfiability}\xspace}
\newcommand*{\EnumWANS}{\textsc{Enumerate Maximal Satisfying \wans Assignments}\xspace}
\newcommand*{\wans}{\textsc{WA$3$NS}\xspace}
\newcommand*{\WS}[1]{\textsc{Weighted $#1$-normalized Satisfiability}\xspace}

\journal{ArXiv.}

\begin{document}

\begin{frontmatter}

\title{The Complexity of Dependency Detection and Discovery in Relational Databases\tnoteref{t1}}

\tnotetext[t1]{An extended abstract of this work was presented at the
11th International Symposium on Parameterized and Exact Computation (IPEC 2016)~\cite{Blaesius16DependencyDetection}.}

\date{}

\author[1]{Thomas Bl{\"a}sius\fnref{fn1}}
\ead{thomas.blaesius@kit.edu}

\affiliation[1]{organization={Karlsruhe Institute of Technology},
            city={Karlsruhe},
            country={Germany}}

\fntext[fn1]{This work originated while all authors were affiliated 
with the Hasso Plattner Institute at the University of Potsdam.}

\author[2]{Tobias Friedrich\fnref{fn1}}
\ead{tobias.friedrich@hpi.de}

\author[2]{Martin Schirneck\corref{cor1}\fnref{fn1}}
\ead{martin.schirneck@hpi.de}

\cortext[cor1]{The third author is corresponding.}

\affiliation[2]{organization={Hasso Plattner Institute, University of Potsdam},
            city={Potsdam},
            country={Germany}}

\begin{abstract}
	Multi-column dependencies in relational databases come associated
	with two different computational tasks.
	The detection problem is to decide whether a dependency of a certain type and size
	holds in a given database,
	the discovery problem asks to enumerate all valid dependencies of that type.
	We settle the complexity of both of these problems for
	unique column combinations (UCCs), functional dependencies (FDs),
	and inclusion dependencies (INDs).
	
	We show that the detection of UCCs and FDs is W[2]-complete
	when parameterized by the solution size.
	The discovery of inclusion-wise minimal UCCs is proven 
	to be equivalent under parsimonious reductions
	to the transversal hypergraph problem
	of enumerating the minimal hitting sets of a hypergraph.
	The discovery of FDs is equivalent
	to the simultaneous enumeration of the hitting sets of multiple input hypergraphs.
	
	We further identify the detection of INDs as 
	one of the first natural \mbox{W[3]-complete} problems.
	The discovery of maximal INDs 
	is shown to be equivalent to enumerating the maximal satisfying assignments of
	antimonotone, 3-normalized Boolean formulas.
\end{abstract}


\end{frontmatter}

\setcounter{footnote}{0}

\section{Introduction}
\label{sec:intro}

\noindent
Data profiling is the extraction of metadata from databases.
An important class of metadata in the relational model are multi-column dependencies.
They describe interconnections between the values stored for different attributes
or even across multiple instances.
Arguably the most prominent type of such dependencies are the unique column combinations (UCCs),
also known as candidate keys.
These are collections of attributes such that the value combinations appearing
in those attributes identify all rows of the database.
In \Cref{fig:example_table}, the \texttt{Name} and \texttt{Area Code}
provide a unique fingerprint of the first database as, for example,
there is only one \texttt{Doe,$\,$John} living in \texttt{UK-W1K}.
Note that none of the two columns can be left out due to duplicate values.
Small UCCs are natural candidates for primary keys,
avoiding the need to introduce surrogate identification.
More importantly though, knowledge of the inclusion-wise minimal unique column combinations
enable various data cleaning tasks as well as query optimization.
For example, value combinations of UCCs are by definition distinct and form groups of size~$1$,
thus SQL queries working on UCCs can skip the grouping phase
and the \texttt{DISTINCT} operation, even if requested by the user~\cite{Chaudhuri94GroupByQueryOptimization,Paulley94ExploitingUniquenessinQueryOptimization}.
Also, the presence of UCCs allows for early returns of \texttt{SELECT}
and \texttt{ORDER BY} operations.

Unfortunately, a dataset only rarely comes annotated with its dependencies.
Much more often they need to be computed from raw data.
This leads to two different computational tasks.
The \emph{detection problem}, is to decide whether for a given database
whether it admits a UCC with only a few columns.
The \emph{discovery problem} instead asks for a complete list of all minimal UCCs,
regardless of their size and number.
An equivalent term for the latter, which is probably more common in the algorithms community,
is the \emph{enumeration} of minimal solutions.

\begin{figure}
\centering
	\includegraphics[scale=0.85]{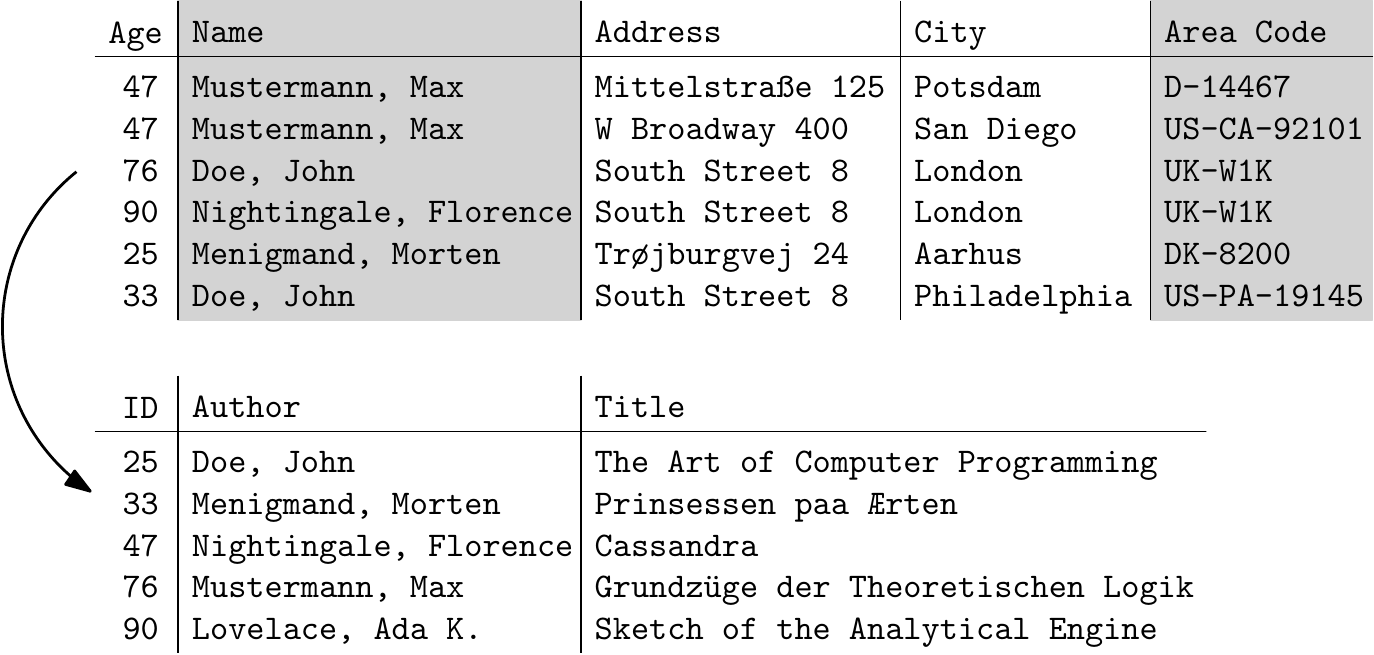}
	\caption{Illustration of multi-column dependencies.
		\texttt{Name} and \texttt{Area$\,$Code} together
		are a minimal unique column combination in the first database.
		The functional dependency $\texttt{Name}, \texttt{Area\,Code} \fd \texttt{City}$ holds,
		but its left-hand side is not minimal since $\texttt{Area\,Code} \fd \texttt{City}$
		is also valid.
		There are two maximal inclusion dependencies of size 1 
		between the first and the second database,
		\texttt{Age} is included in \texttt{ID} and \texttt{Name} in \texttt{Author}.
		They cannot be combined to an inclusion dependency of size 2.}
\label{fig:example_table}
\end{figure}

Functional dependencies (FDs), another dependency type, model the case in which one is only interested
in identifying the values of a specific column, instead of all columns.
In the example in \Cref{fig:example_table},
it is enough to know the \texttt{Area Code} to also infer the \texttt{City}
since the value in the former determines the latter.
Finally, inclusion dependencies (INDs) reveal connections between different databases
as they describe the fact that all value combinations in certain columns of one database 
also appear for the attributes of another.
In contrast to UCCs and FDs, where we want solutions to be small,
here we ought to find large or maximal inclusion dependencies.
Those are much more likely to be caused by the inherent structure of the data
than by mere coincidence.
Functional dependencies and inclusion dependencies
can be used for example for cardinality estimation in query plan optimizers,
query rewriting, and joins~\cite{Casanova82INDsandFDs,Giannella02QueryEvalApproximateFDs,Ilyas04CORDS_SIGMOD}.

The detection (decision) problems for all three types of dependencies are \NP-complete.\footnote{%
	See \Cref{subsec:UCC_problem_definitions,subsec:IND_problem_definitions}
	for specific statements and references.
}
Notwithstanding, detection algorithms often perform well on practical datasets~\citep{Abedjan18DataProfiling}.
One approach to bridge this apparent gap is to analyze whether properties
that are usually observed in realistic data benevolently influence the hardness of the problem. 
Exploiting those properties may even lead to algorithms
that guarantee a polynomial running time
in case these features are present in the problem instance.
This is formalized in the concept of parameterized algorithms~\cite{Cygan15ParamertizedAlgorithms,DowneyFellows13Parameterized,Niedermeier06Invitation}.
There, one tries to ascribe the exponential complexity
entirely to an parameter of the input, other than its mere size.
If the parameter is bounded for a given class of instances,
we obtain a polynomial running time whose order of growth is independent of the parameter.
This, of course, requires one to identify parameters 
that are both algorithmically exploitable and small in practice.

Consider, for instance, the histograms in \Cref{fig:histogram}, showing the
size distribution of minimal unique column combinations and
functional dependencies, as well as maximal inclusion dependencies in the
MusicBrainz
database~\cite{Swartz02MusicBrainz}.
The majority of functional dependencies are rather small,
same for unique column combinations and inclusion dependencies.
Beside surrogate keys, giving rise to
multiple functional dependencies of size~1, causalities in the data can also
lead to small FDs.  For example, the name of an
event together with the year in which it started determines the year in
which it ends, implying an FD of size~2.  Note
that the starting year alone is usually not enough to infer this information.
The name of the action, however, seems to indicate whether the event
ends in the same year or the next.
The size of the dependency is thus a natural candidate for an algorithmic parameter.
Notwithstanding, we show that it is unlikely to be the sole explanation for
the good practical performance.
We prove that the detection of unique column combinations
and functional dependencies is $\W[2]$-complete
with respect to the size of the sought solution, detecting
inclusion dependencies is even $\W[3]$-complete.
For all we know, this excludes any algorithm parameterized by the size.

\begin{figure}
  \centering
  \includegraphics{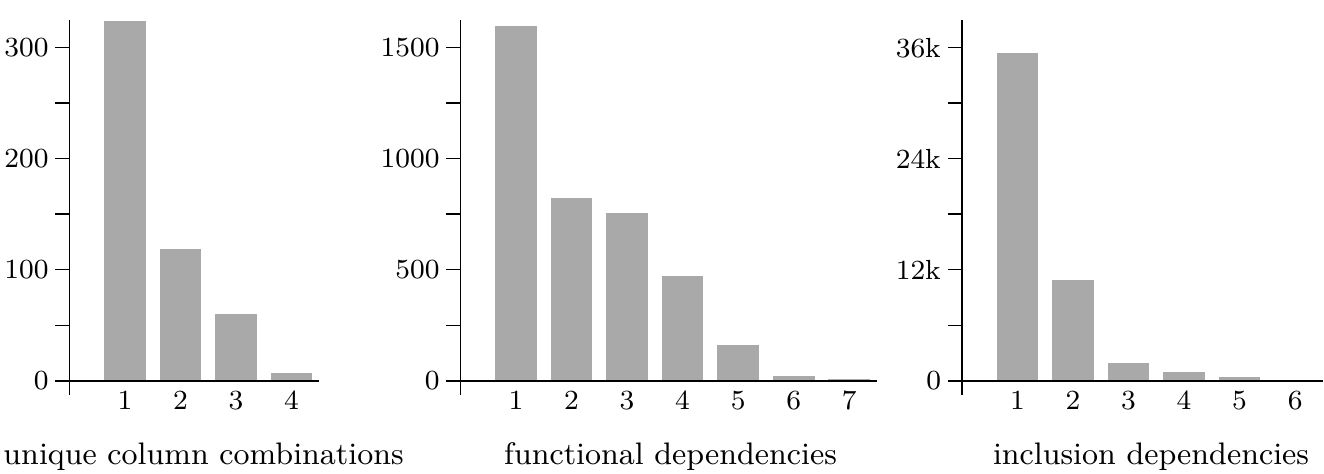}
  \caption{The number of minimal unique column combinations, minimal
    functional dependencies, and maximal inclusion dependencies
    for varying solution sizes in the MusicBrainz database.}
  \label{fig:histogram}
\end{figure}

The hardness of detecting INDs is surprising also from a complexity-theoretic standpoint.
Currently, there are only a handful of natural problems known to be complete for the class $\W[3]$.
The first one was given by Chen and Zhang~\cite{Chen06W3W4} in the context of supply chain management.
We show here that the detection of inclusion dependencies has this property,
making it the second such problem.
Since this result was first announced, Bläsius~et~al.~\cite{Blaesius19EfficientlyALENEX}
have also proven the extension problem for minimal hitting sets to be $\W[3]$-complete
using different techniques.
The latter has subsequently been improved by Casel et al.~\cite{Casel18ComplexitySolutionExtensionArxiv},
they have shown that already the special case of extension to minimal dominating sets
in bipartite graphs is hard for $\W[3]$. 
Finally, building on the work presented here,
Hannula, Song, and Link~\cite{Hannula21IndependenceFromDataArXiv}
have very recently identified independence detection in relational databases
as another representative of this class.

For enumeration problems,
the border of tractability does not run between polynomial and super-polynomial time,
at least not when measured in the input only.
The number of solutions one wishes to be computed may be exponential in the input size,
ruling out any polynomial algorithm.
Instead, one could ask for an algorithm 
that scales polynomially both in the input and the number of solutions.
The most important yardstick in enumeration complexity is the \emph{transversal hypergraph} problem,
where one is tasked to compute all minimal hitting sets of a given hypergraph.
Currently, the fastest known algorithm runs in time $N^{\Or(\sfrac{\log N}{\log\log N})}$,
where $N$ is the combined input and output size~\cite{FredmanKhachiyan96Dualization}.
It is the major open question in the field of enumeration algorithms
whether the transversal hypergraph problem can be solved in output-polynomial time.
Besides from data profiling, the problem emerges in many applications in
artificial intelligence, machine learning, distributed systems, monotone logic, and bioinformatics
see~\cite{EiterMakinoGottlob08Survey,Pusa20MOOMIN}.
There are many algorithms known that work well on practical instances~\cite{Gainer17AlgorithmsAndComputation}.

We use the insights gained on the detection of multi-column dependencies in databases
to also investigate their discovery (enumeration) problems.
It is known that minimal unique column combinations and functional dependencies
can be discovered in output-polynomial time if and only if
the transversal hypergraph problem has an output-poly\-no\-mial solution~\cite{EiterGottlob95RelatedProblems}.
However, this was proven via a Turing-style reduction that continuously calls
a decision subroutine to check 
whether the enumeration part has already found all solutions.
This construction inherently uses space
proportional to the output and is therefore hardly useful in practice.
We are able to radically simplify this equivalence
and connect unique column combinations, functional dependencies, and hitting sets
directly at the enumeration level using so-called \emph{parsimonious reductions},
running in polynomial time and space.
We give similar results also for the discovery of maximal inclusion dependencies.

Parsimonious reductions are the most restrictive form of reduction between enumeration problems,
but--therefore--also the most useful in practice.
The close connection to the transversal hypergraph problem that we prove in this work
explains in parts why dependency discovery works quite well on real-world databases.
Moreover, it allows us to transfer ideas from the design of hitting set enumeration 
algorithms to data profiling, thereby connecting the two research areas.
For example, there are very space-efficient algorithms known for the transversal hypergraph problem,
while memory consumption still seems to be a major obstacle in dependency discovery~\cite{Papenbrock16HyFD,Wei19DiscoveryEmbeddedUnique}.  
We hope that our results here inspire better data profiling algorithms in the future.\\

\noindent
\textbf{Our Contribution.} 
We settle the parameterized complexity of the car\-di\-na\-lity-constrained decision problems
for unique column combinations, functional dependencies, and inclusion dependencies
in relational databases, with the solution size as parameter.
We prove the following theorems.

\begin{thm}
\label{thm:unique_fd_W[2]}
	Detecting a unique column combination of size $k$ in a relational database
	is $\emph{\W}[2]$-complete when parameterized by $k$.
	The same is true for the detection of a valid, non-trivial functional dependency
	with a left-hand side of size at most $k$, 
	even if the desired right-hand side is given in the input.
\end{thm}

\begin{thm}
\label{thm:ind_W[3]}
	Detecting an inclusion dependency of size $k$ in a pair of relational databases
	is $\emph{\W}[3]$-complete when parameterized by $k$.
	The result remains true even if both databases are over the same relational schema
	with the identity mapping between their columns.
\end{thm}

We also characterize the complexity of enumerating all multi-column dependencies
of a certain type in a given database.
We do so by proving parsimonious equivalences with well-known enumeration problems
as well as generalizations thereof.

\begin{thm}
\label{thm:equiv_transHyp}
	The following enumeration problems are equivalent under parsimonious reductions:
	\begin{enumerate}
		\item listing all minimal unique column combinations of a relational database;
		\item listing all minimal, valid, and non-trivial functional dependencies
			with a fixed right-hand side;
		\item the transversal hypergraph problem.
	\end{enumerate}		
	The enumeration of functional dependencies with arbitrary right-hand sides
	is equivalent to listing the hitting sets of multiple input hypergraphs.
	The latter two problems are at least as hard as the transversal hypergraph problem.
\end{thm}

\begin{thm}
\label{thm:enum_ind}
	The following enumeration problems are equivalent under parsimonious reductions:
	\begin{enumerate}
		\item listing all maximal inclusion dependencies of a pair of relational\\
				databases;
		\item listing all maximal satisfying assignments of an antimonotone,\\
				3-normalized Boolean formula.
	\end{enumerate}
	This remains true even if the two databases are over the same schema
	and only inclusions between the same columns are allowed.
	All those enumeration problems are at least as hard as the transversal hypergraph problem.
\end{thm}

Finally, we also briefly discuss the consequences of our findings
to the approximability of minimum dependencies.\vspace*{.5em}

\noindent
\textbf{Outline.}
In the next section, we fix some notation and review basic concepts needed for the later proofs.
\Cref{sec:UCCs_and_FDs} treats unique column combinations and functional dependencies,
\Cref{sec:INDs} then considers inclusion dependencies.
Both of these sections start with a segment
that formally defines all decision and enumeration problems discussed in the respective section.
The paper is concluded in \Cref{sec:conclusion}.

\section{Preliminaries}
\label{sec:prelims}

\subsection{Hypergraphs and Hitting Sets}
\label{subsec:prelims_hypergraphs}

\noindent
A \emph{hypergraph} is a non-empty, finite \emph{vertex set} $V \neq \emptyset$ together
with a set of subsets $\Hyp \subseteq \mathcal{P}(V)$, the \emph{(hyper-)edges}.
We identify a hypergraph with its edge set $\Hyp$ if this does not create any ambiguities.
We do not exclude special cases of this definition like the empty graph
($\Hyp = \emptyset$), an empty edge ($\emptyset \in \Hyp$),
or isolated vertices ($V \supsetneq \bigcup_{E \in \Hyp} E$).
A hypergraph is \emph{Sperner} if none of its edges is contained in another.
The \emph{minimization} of $\Hyp$ is the subset of all
inclusion-wise minimal edges, 
$\min(\Hyp) = \{ E \in \Hyp \mid \forall E' \in \Hyp \colon E' \subseteq E \Rightarrow E' = E\}$.
This should not be confused with the notation for a minimum element of a set.
The minimization of a hypergraph is always Sperner.

A \textit{hitting set}, or \textit{transversal}, of a hypergraph $(V, \Hyp)$
is a set $T \subseteq V$ of vertices 
such that $T$ has a non-empty intersection with every edge $E \in \Hyp$.
A hitting set is \textit{(inclusion-wise) minimal}
if it does not properly contain another hitting set.
The minimal hitting sets of $\Hyp$
form a Sperner hypergraph on the same vertex set $V$,
the \textit{transversal hypergraph} $\Tr(\Hyp)$.
Regarding transversals, it does not make a difference whether
the full hypergraph $\Hyp$ is considered or its minimization,
it holds that $\Tr(\min(\Hyp)) = \Tr(\Hyp)$.

\subsection{Relational Databases and Dependencies}
\label{subsec:prelims_databases}

\noindent
A \emph{(relational) schema} $R$ is a non-empty, finite set of \emph{attributes} or \emph{columns}.
Each attribute comes implicitly associated with a set of admissible values.
A \emph{row}, or \emph{record}, over schema $R$ is a tuple $r$
whose entries are indexed by $R$ such that,
for any attribute $a \in R$, the value $r[a]$ is admissible for $a$.
A \emph{(relational) database} $\rel$ over $R$ is a finite set of records.
For any row $r$ and subset $X \subseteq R$ of columns,
$r[X]$ is the subtuple of $r$ projected onto $S$.
We let $\rel[X]$ denote the family of all such subtuples of rows in $\rel$.
Note that $\rel[X]$ is a multiset
as the same combination of values may appear in different rows.

We are interested in multi-column dependencies in relational databases,
namely, unique column combinations, functional dependencies, and inclusion dependencies.
For their definition, we need the notion of difference sets.
For any two distinct rows $r_1, r_2$, $r_1 \neq r_2$, over the same schema $R$,
their \emph{difference set} $\{a \in R \mid r_1[a] \neq r_2[a]\}$ 
is the set of attributes in which the rows disagree.
A difference set is \emph{(inclusion-wise) minimal} if it does not properly contain
a difference set for another pair of rows in the database.
We denote the hypergraph of minimal difference sets by $\Dyp$.
The vertex set is $R$.

A \emph{unique column combination} (UCC), or simply \textit{unique},
for some database $\rel$ over schema $R$,
is a subset $U \subseteq R$ of attributes 
such that for any two records $r, s \in \rel$, $r \neq s$,
there is an attribute $a \in U$ such that $r[a] \neq s[a]$.
Equivalently, $U$ is a UCC for $\rel$
iff the collection $\rel[U]$ of projections onto $U$ is mere set.
A UCC is \textit{(inclusion-wise) minimal} if it does not properly contain another UCC.
There is an intimate connection between UCCs and transversals in hypergraphs.
It is well-known in the literature,
see \cite{Abedjan18DataProfiling,Date03DatabaseSystems},
probably the first explicit mention was by Mannila and Räihä~\cite{Mannila87Dependency}.

\begin{obs}[Folklore]
\label{obs:UCCs_and_transversals}
	The unique column combinations are the hitting sets of difference sets.	
	In particular, let $\rel$ be a database and $\Dyp$
	the hypergraph of its minimal difference sets.
	Then, any minimal transversal in $\Tr(\Dyp)$ is a minimal unique of $\rel$
	and there are no other minimal uniques in $\rel$.
\end{obs}

\emph{Functional dependencies} (FDs) over a schema $R$
are expressions of the form $X \fd a$ for some set $X \subseteq R$ of
columns and a single attribute $a \in R$. 
The set $X$ is the \emph{left-hand side} (LHS) of the dependency
and $a$ the \emph{right-hand side} (RHS).
We say that the FD has \emph{size} $|X|$.
A functional dependency $X \fd a$ is said to \emph{hold},
or be \emph{valid}, in an database $\rel$ (over $R$)
if any pair of records that agree on $X$ also agree on $a$,
that is, if $r[X] = s[X]$ implies $r[a] = s[a]$ for any $r,s \in \rel$.
Otherwise, $X \fd a$ is said to \emph{fail} in $\rel$, or be \emph{invalid}.
The FD $\emptyset \fd a$ holds iff all rows agree on $a$.
An FD $X \fd a$ is \emph{(inclusion-wise) minimal} if it holds in $\rel$ and
$X'' \fd a$ fails for any proper subset $X'' \subsetneq X$.
A functional dependency is \emph{non-trivial} if $a \notin X$.
Note that trivial functional dependencies hold in any database.

Finally, we turn to multi-column dependencies among multiple databases.
Let $R$ and $S$ be two relational schemas and $\rel$ and $\sel$
databases over $R$ and $S$, respectively.
For some $X \subseteq R$, let
$\sigma \colon X \to S$ be an injective map.  The pair
$(X, \sigma)$ is an \emph{inclusion dependency} (IND) if, for each
row $r \in \rel$, there exists some $s \in \sel$ such that
$r[a] = s[\sigma(a)]$ for every $a \in X$, that is, 
iff the inclusion $\rel[X] \subseteq \sel[\sigma(X)]$ holds. 
If the map $\sigma$ is given in the input,
we say that $X$ is the dependency.
Such an inclusion dependency then is \emph{maximal}
if the set $X$ is maximal among all INDs for $\rel$ and $\sel$.
For the general case, we define a partial order on the pairs $(X,\sigma)$.
We say that $(X,\sigma) \preccurlyeq (X',\sigma')$ holds 
if $X \subseteq X'$ and $\sigma$ is the restriction of $\sigma'$ to $X$.
An inclusion dependency then is \emph{maximal}
if it is an $\preccurlyeq$-maximal element among the inclusion dependencies
between $\rel$ and $\sel$.
Observe that the inclusion dependencies are indeed downward closed with respect to $\preccurlyeq$.
However, it may happen that $(X,\sigma)$ and $(X',\sigma')$ are both maximal
although $X' \subsetneq X$ is a strict subset.

\subsection{Parameterized Complexity}
\label{subsec:prelims_parameterized}

\noindent
The central idea of the \emph{parameterized complexity} of a decision problem
is to identify a quantity of the input, other than its mere size,
that captures the hardness of the problem. 
The decision problem associated with a language $\mathrm{\Pi} \subseteq \{0,1\}^*$
is to determine whether some instance $I \in \{0,1\}^*$ is in $\mathrm{\Pi}$.
A decision problem is \emph{parameterized} if any instance $I$ is
additionally augmented with a \emph{parameter} $k = k(I) \in \N^+$,
we thus have $\mathrm{\Pi} \subseteq \{0,1\}^* \,{\times}\, \N^+$.
A parameterized decision problem $\mathrm{\Pi}$ is \emph{fixed-parameter tractable} (FPT),
if there exists a computable function $f \colon \N^+ \to \N^+$ and an algorithm
that decides any instance $(I,k)$ in time $f(k) \,{\cdot}\, \poly(|I|)$.
The complexity class \FPT
is the collection of all fixed-parameter tractable problems.

Let $\mathrm{\Pi}$ and $\mathrm{\Pi}'$ be two parameterized problems.
A \emph{parameterized reduction}, or \emph{FPT-reduction},
from $\mathrm{\Pi}$ to $\mathrm{\Pi}'$ is an algorithm
running in time $f(k)\cdot \poly(|I|)$ on instances $(I, k)$,
which outputs some instance $(I', k')$ 
such that $k' \le g(k)$ for some computable function $g \colon \N^+ \to \N^+$,
and $(I,k) \in \mathrm{\Pi}$ holds if and only if $(I',k') \in \mathrm{\Pi}'$ does.
Due to the time bound, we have $|I'| \le f(k) \, \poly(|I|)$.
If there is also an FPT-reduction from $\mathrm{\Pi}'$ to $\mathrm{\Pi}$,
we say that the problems are \emph{FPT-equivalent}.
A notable special case of FPT-reductions are
polynomial many-one reductions that preserve the parameter,
meaning $k' = k$.

Parameterized reductions give rise to the so-called \W-\emph{hierarchy} of complexity classes.
We use one of several equivalent definitions involving (mainly) Boolean formulas.
However, first consider the \textsc{Independent Set} problem on graphs
parameterized by the size of the sought solution.
The class $\W[1]$ is the collection of all parameterized problems that admit
a parameterized reduction to \textsc{Independent Set}.
For some positive integer $t$, a Boolean formula is \emph{$t$-normalized}
if it is a conjunction of disjunctions of conjunctions of disjunctions (and so on) of literals
with $t\nwspace{-}\nwspace{1}$ alternations or, equivalently,
$t$ levels in total.
The \WS{t} problem is to decide, for a $t$-normalized formula $\varphi$ and a positive integer $k$,
whether $\varphi$ admits a satisfying assignment of Hamming weight $k$,
that is, with (exactly) $k$ variables set to \true.
Here, $k$ is the parameter.
For every $t \ge 2$, $\W[t]$ is the class of all problems reducible to \WS{t}.\footnote{%
	The definition via normalized formulas comes with an inconsistency at $t=1$.
	A \mbox{$1$-normalized} formula is a single conjunctive clause,
	the associated weighted satisfiability problem is trivially seen to be in $\P$
	and thus in $\FPT$.
}
The classes
$\FPT \subseteq \W[1] \subseteq \W[2] \subseteq \dots$
form an ascending hierarchy.
All inclusion are conjectured to be strict,
which is however still unproven.

\subsection{Enumeration Complexity}
\label{subsec:prelims_enumeration}

\noindent
\emph{Enumeration} is the task of compiling and outputting a list of 
all solutions to a computational problem without repetitions.
Note that this is different from a counting problem,
which asks for the mere number of solutions.
More formally, an \emph{enumeration problem} is a function
$\mathrm{\Pi} \colon \{0,1\}^* \to \mathcal{P}(\{0,1\}^*)$
such that, for all instances $I \in \{0,1\}^*$, 
the set of \emph{solutions} $\mathrm{\Pi}(I)$ is finite.
An algorithm solving this problem needs to output, on input $I$,
all elements of $\mathrm{\Pi}(I)$ exactly once.
We do not impose any order on the output.
We focus on the enumeration of minimal hitting sets, that is,
$\mathrm{\Pi} \colon (V,\Hyp) \mapsto \Tr(\Hyp)$.

An \emph{output-polynomial} algorithm runs in time polynomial
in the combined input and output size $N = |V| + |\Hyp| + |\Tr(\Hyp)|$.
A (seemingly) stronger requirement is an \emph{incremental polynomial}
  algorithm, generating the solutions in such a way that the $i$-th
\emph{delay}, the time between the $(i\nwspace{-}\nwspace{1})$-st and $i$-th output, is
bounded by $\poly(|V|,|\Hyp|,i)$.
This includes the \emph{preprocessing time} until the first solution arrives ($i = 1$)
as well as the \emph{postprocessing time} between the last solution and termination
($i = |\Tr(\Hyp)|\,{+}\,1$).
The strongest form of
output-efficiency is that of \emph{polynomial delay}, where the
delay is universally bounded by $\poly(|V|,|\Hyp|)$.
There is currently no output-polynomial algorithm known for the transversal hypergraph problem,
but its existence would immediately imply also an incremental polynomial algorithm~\cite{Bioch95Identification}.

Arguably the most restrictive way to relate enumeration problems are parsimonious reductions.
The concept is closely related but should not be confused with 
the homonymous class of reductions for counting problems~\cite{CapelliStrozecki19Incremental}.
A \emph{parsimonious reduction} from enumeration problem $\mathrm{\Pi}$ to $\mathrm{\Pi}'$
is a pair of polynomial time computable functions $f \colon \{0,1\}^* \to \{0,1\}^*$
and $g \colon (\{0,1\}^*)^2 \to \{0,1\}^*$ such that, for any instance $I \in \{0,1\}^*$,
$g(I, \cdot)$ is a bijection from $\mathrm{\Pi}'(f(I))$ to $\mathrm{\Pi}(I)$.
The behavior of $g(I, \cdot)$ on $\{0,1\}^*{\setminus}\mathrm{\Pi}'(f(I))$ is irrelevant.
Intuitively, any enumeration algorithm for $\mathrm{\Pi}'$ can then be turned
into one for $\mathrm{\Pi}$ by first mapping the input $I$ to $f(I)$
and translating the solutions back via $g$.
\Cref{obs:UCCs_and_transversals} establishes
a parsimonious reduction from the enumeration of minimal UCCs 
to the transversal hypergraph problem with $f \colon (R,\rel) \mapsto (R,\Dyp)$ and
$g((R,\rel), \cdot)$ being the identity over subsets of $R$.

\section{Unique Column Combinations and Functional Dependencies}
\label{sec:UCCs_and_FDs}

\noindent
Theoreticians as well as practitioners in data profiling and database design
are frequently confronted with the task of finding a small collection of items
that has a non-empty intersection with each member of a prescribed family of sets,
see~\cite{Abedjan18DataProfiling,Date03DatabaseSystems,Davies94SmallestFeatureSets,
Kantola92FDsAndINDs,Maier83TheoryRelationalDatabases}.
They thus aim to solve instances of the hitting set problem.
In this section, we show that this encounter is inevitable in the sense that 
detecting a small unique column combination
or functional dependency in a relational database is the same
as finding a hitting set in a hypergraph.
Even more so, this equivalence extends to enumeration.
We show that the associated discovery problems of finding all UCCs or FDs
is indeed the same as enumerating all hitting sets.

We first formally define the respective decision and enumeration problems.
The decision versions are always parameterized by the solution size.
We then order them in a (seemingly) ascending chain via parameterized reductions.
However, the start and end points of this chain will turn out to be FPT-equivalent,
settling the complexity of the problems involved as complete for the 
parameterized complexity class $\W[2]$.
We then also show the equivalence of the corresponding enumeration problems under parsimonious reductions.

\subsection{Problem Definitions}
\label{subsec:UCC_problem_definitions}

\noindent
Recall the definitions of hitting sets as well as
unique column combinations and functional dependencies from \Cref{subsec:prelims_hypergraphs,subsec:prelims_databases}.
We are interested in the parameterized complexity
of the associated cardinality-constrained decision problems.
The constraint always serves as the parameter.

\vspace*{.6em}
\noindent\HS
\begin{description}
	\item [Instance:] A hypergraph $(V,\Hyp)$ and a non-negative integer $k$.
	\item [Parameter:] The non-negative integer $k$.
	\item [Decision:] Is there a set $T \subseteq V$ of vertices with cardinality $|T| = k$
	
				such that $T$ is a hitting set for $\Hyp$?
\end{description}

\noindent
Note that if $k > |V|$, then the answer to the decision problem is trivially \false;
otherwise, there is no difference between deciding the existence of a transversal
with \emph{at most} or \emph{exactly} $k$ elements
since every superset of a hitting set is again a hitting set.
We ignore the special case of a too large $k$ since parameterized complexity is primarily concerned
with the situation that the parameter is much smaller than the input size.
The unparameterized \HS problem is one of Karp's initial 21 \mbox{\NP-complete} problems~\cite{Karp72Reducibility}.
In fact, its minimization variant--to compute the minimum size of any hitting set--is
even \NP-hard to approximate within a factor of $(1 - \varepsilon) \ln |V|$
for any \mbox{$\varepsilon > 0$}~\cite{Dinur14ParallelRepetition}.
The parameterized variant defined above is the prototypical \mbox{$\W[2]$-complete} problem~\cite{DowneyFellows13Parameterized}.

The corresponding enumeration problem broadens the notion of a ``small'' solution,
namely, the task is now to list all \emph{inclusion-wise minimal} hitting sets,
that is, the edges of the transversal hypergraph $\Tr(\Hyp)$.
All other hitting sets can be trivially obtained from the minimal ones
by arbitrarily adding more vertices.

\vspace*{.6em}
\noindent\TransHyp
\begin{description}[labelwidth = 2.45cm]
	\item [Instance:] A hypergraph $(V,\Hyp)$.
	\item [Enumeration:] List all edges of $\Tr(\Hyp)$.
\end{description}

\noindent
Let $N = |\Hyp| + |\Tr(\Hyp)| + |V|$ denote the combined input and output size,
Fredman and Kachiyan's algorithm solves \TransHyp in time 
$N^{\Or(\sfrac{\log N}{\log\log N})}$~\cite{FredmanKhachiyan96Dualization}.

We generalize the problem to the enumeration of minimal hitting sets
for multiple input hypergraphs simultaneously.
We do not prescribe any order in which the hypergraphs are processed.
However, we want to be able to quickly tell to which input a solution belongs.
For this, let $\Tr(\Hyp) \nwspace\dot\cup\nwspace \Tr(\Gyp)$
denote the disjoint union of the transversal hypergraphs of $(V, \Hyp)$ and $(W,\Gyp)$
with the additional requirement that the union is encoded in such a way that,
for any $T \in \Tr(\Hyp) \nwspace\dot\cup\nwspace \Tr(\Gyp)$,
the containment $T \in \Tr(\Hyp)$
is decidable in time $\poly(|V|,|W|,|\Hyp|, |\Gyp|)$,
independently of the sizes $|\Tr(\Hyp)|$ and $|\Tr(\Gyp)|$.

\vspace*{.6em}
\noindent\TransHypUnion
\begin{description}[labelwidth = 2.45cm]
	\item [Instance:] A $d$-tuple of hypergraphs $(\Hyp_1, \Hyp_2, \dots, \Hyp_d)$.
	\item [Enumeration:] List all edges of 
		$\Tr(\Hyp_1) \,\dot\cup\, \Tr(\Hyp_2) \,\dot\cup\, \dots \,\dot\cup\, \Tr(\Hyp_d)$.
\end{description}

We now define the detection and discovery problems of multi-column dependencies 
in relational databases, again starting with the cardinality-constraint decision problems.

\vspace*{.6em}
\noindent\Unique
\begin{description}
	\item [Instance:] A relational database $\rel$ over schema $R$ 
	
		and a non-negative integer $k$.
	\item [Parameter:] The non-negative integer $k$.
	\item [Decision:] Is there a set $U \subseteq R$ of attributes with cardinality $|U| = k$
	
		such that $U$ is a unique column combination in $\rel$?
\end{description}

\noindent
The minimization variant of \Unique is \NP-hard~\cite{Skowron92DiscernibilityMatrices}.

\vspace*{.6em}
\noindent\EnumUnique
\begin{description}[labelwidth = 2.45cm]
	\item [Instance:] A relational database $\rel$.
	\item [Enumeration:] List all minimal unique column combinations of $\rel$.
\end{description}

For functional dependencies,
we define two variants of the decision problem that slightly differ in the given input.
The first one fixes the right-hand side of the desired dependency,
while the second one asks for an FD with arbitrary RHS holding in the database.
While their parameterized complexity will turn out to be the same,
there are some differences in their discovery.

\pagebreak

\noindent\FDfixed
\begin{description}
	\item [Instance:] A relational database $\rel$ over schema $R$, an attribute $a \in R$,
		
			and a non-negative integer $k$.
	\item [Parameter:] The non-negative integer $k$.
	\item [Decision:] Is there a set $X \subseteq R{\setminus}\{a\}$ with $|X| = k$ such that
	
			the functional dependency $X \fd a$ holds in $\rel$?
\end{description}

\vspace*{.6em}
\noindent\FD
\begin{description}
	\item [Instance:] A relational database $\rel$ over schema $R$
	
		and a non-negative integer $k$.
	\item [Parameter:] The non-negative integer $k$.
	\item [Decision:] Is there a set $X \subseteq R$ with $|X| = k$ 
		and an attribute $a \in R{\setminus}X$
	
		such that the functional dependency $X \fd a$ holds in $\rel$?
\end{description}

\noindent
The unparameterized variant of \FDfixed is \NP-complete
even if the number of admissible values for each attribute is at most~$2$~\cite{Davies94SmallestFeatureSets}.
It is the same to ask for a functional dependency whose left-hand side
is of size at most $k$;
unless of course $k \ge |R|$ 
since then no non-trivial FD adheres to the (exact) size constraint.
Again, we ignore this issue.

Recall that we say that a functional dependency $X \fd a$ is minimal
if its LHS $X$ is inclusion-wise minimal among all $X' \subseteq R$ such that $X' \fd a$ is valid.

\vspace*{.6em}
\noindent\EnumFDfixed
\begin{description}[labelwidth = 2.45cm]
	\item [Instance:] A relational database $\rel$ over schema $R$
		and an attribute $a \in R$.
	\item [Enumeration:] List all minimal, valid, non-trivial functional dependencies of~$\rel$
		
		\hspace*{0.3cm} with right-hand side $a$.
\end{description}

\vspace*{.6em}
\noindent\EnumFD
\begin{description}[labelwidth = 2.45cm]
	\item [Instance:] A relational database $\rel$.
	\item [Enumeration:] List all minimal, valid, non-trivial functional dependencies~of~$\rel$.
\end{description}

\noindent
\EnumUnique and \EnumFD can be solved in output-polynomial time
(in incremental polynomial time or with polynomial delay, respectively)
if and only if this is possible for the \TransHyp problem~\cite{EiterGottlob95RelatedProblems}.

\subsection{Detection}
\label{subsec:UCC_detection}

\noindent
Next, we prove the parameterized complexity
of the detection problems for unique column combinations
and functional dependencies.
We start by showing the $\W[2]$-hardness of \Unique.
The following lemma was obtained independently by Froese et al.~\cite{Froese16DistinctVectors}.

\begin{figure}
    \centering
    \begin{subfigure}{0.5\textwidth}
        \centering
        \includegraphics[page=1,scale=1.1]{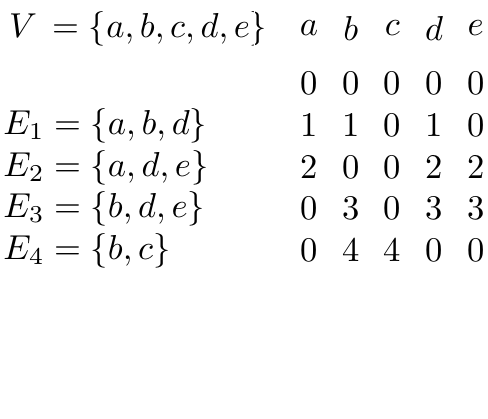}
        \captionsetup{justification=centering}
		\caption{}													
	\label{subfig:unique-reduction}
    \end{subfigure}%
    \begin{subfigure}{0.5\textwidth}
        \centering    
	    \includegraphics[page=2,scale=1.1]{unique-fd-reductions}
	    \captionsetup{justification=centering}
		\caption{}													
	\label{subfig:fd-reduction}
    \end{subfigure}
    \caption{Illustration of the reductions in \Cref{lem:hs_to_unique,lem:unique_to_fd}.
		(a)~An instance of \HS and the equivalent instance of \Unique.
		(b)~An instance $\rel$ of \FDfixed with fixed right-hand side $a$
		and the equivalent instance $\rel'$ of \FD.
		The functional dependencies $ab \fd d$ holds in $\rel$, but not in $\rel'$.}
  \label{fig:unique-fd-reductions}
\label{fig:oracle_run_times}
\end{figure}

\begin{lem}
\label[lem]{lem:hs_to_unique}
	There is a parameterized reduction from \emph{\HS}	to the \emph{\Unique} problem.
\end{lem}

\begin{proof}
	Let $(V,\Hyp)$ be the hypergraph given in the input to the \HS problem.
	Without loss of generality, we can assume it to be Sperner;
	otherwise, we replace it by its minimization $\min(\Hyp)$ (in quadratic time).
	Observe that $\min(\Hyp)$ has a hitting set of size at most $k$ iff $\Hyp$ has one.
	We construct from $\Hyp$ in polynomial time a database $\rel$ over schema $V$
	such that the minimal difference sets of $\rel$ are the edges of $\Hyp$.
	The lemma then follows immediately from \Cref{obs:UCCs_and_transversals}.
	In particular, since \Cref{obs:UCCs_and_transversals}
	transfers solutions, the parameter $k$ is preserved by the reduction.
	
	Let $E_1, E_2, \dots, E_m$ be the edges of $\Hyp$.
	We take the integers $0,1,\dots,m$ as the admissible values for all attributes in $V$.
	As rows, we first add the all-zero tuple $r_0 = (0)_{a \in V}$ to $\rel$.
 	For each $i \in [m]$, we also add the record $r_i$ defined as
	\begin{equation*}
		r_i[a] =
			\begin{cases}
				i, & \text{if } a \in E_i;\\
				0, & \text{otherwise.}
			\end{cases}
	\end{equation*}
	
	\noindent
	See \Cref{subfig:unique-reduction} for an illustration.
	Clearly, $\rel$ can be computed in linear time.
	
	Any edge $E_i$ is a difference set in $\rel$, namely,
	that of the pair $(r_0, r_i)$.
	Any other difference set must come from a pair $(r_i, r_j)$
	with $1 \le i < j \le m$.
	It is easy to see that those rows disagree in $E_i \cup E_j$,
	which is not minimal.
	Since $\Hyp$ is Sperner, it contains exactly the minimal difference sets of $\rel$.
\end{proof}

The next two reductions are straightforward due
to the similar structures of uniques and functional dependencies.
While a UCC separates any pair of rows,
an FD $X \fd a$ needs to distinguish only those with $r[a] \neq s[a]$.

\begin{lem}
\label[lem]{lem:unique_to_fd}
	There are parameterized reductions
	\begin{enumerate}
		\item\label[stat]{case:unique_to_fdfixed} from \emph{\Unique} 
			
			to \emph{\FDfixed};
		\item\label[stat]{case:fdfixed_to_fd} from \emph{\FDfixed}
		
			to \emph{\FD}.
	\end{enumerate}
\end{lem}

\begin{proof}
	To prove \Cref{case:unique_to_fdfixed}, we add a single unique column to the database
	and fix it as the right-hand side of the sought functional dependency.
	Let $\rel = \{r_1, r_2, \dots, r_{|\rel|}\}$ be a database over schema $R$,
	and $a$ an attribute not previously in $R$.
	We construct $\rel'$ over $R \cup \{a\}$ by adding,
	for each $r_i$, the row $r'_i$ defined by $r'_i[R] = r_i[R]$ and $r'_i[a] = i$.
	The reduction maps an instance $(\rel,R,k)$ of \Unique
	to the instance $(\rel',R \cup \{a\}, a, k)$ of \FDfixed.
	Since, for any two distinct rows $r'_i,r'_j \in \rel'$, $i \neq j$,	we have $r'_i[a] \neq r'_j[a]$,
	the left-hand sides of the non-trivial, valid FDs $X \fd a$ in $\rel'$ are
	in one-to-one correspondence to the UCCs in $\rel$.
	
	To reduce \FDfixed to \FD, we need to mask all ``unwanted'' FDs
	with RHS different from the fixed attribute $a$.
	See \Cref{subfig:fd-reduction} for an example.
	Let again $\rel$ be the input database over $R$.
	To construct the resulting database $\rel'$ over the same schema $R$,
	we take all rows from $\rel$ and add $|R|-1$ new ones.
	Fix an arbitrary record $r^* \in \rel$ and let ${\times}$ be a new symbol
	that does not previously appear as a value.
	For each attribute $b \in R{\setminus}\{a\}$, we add the row $r_b$
	satisfying $r_b[R{\setminus}\{b\}] = r^*[R{\setminus}\{b\}]$ and $r_b[b] = {\times}$.
	The rows $r^*$ and $r_b$ now
	witness that any non-trivial FD $X \fd b$ fails in $\rel'$.
	
 	It is left to prove that $X \fd a$ holds in $\rel$ if and only if it holds in $\rel'$.
 	Evidently, any valid FD in $\rel'$ is also valid in the subset $\rel$.
 	Suppose $X \fd a$ holds in $\rel$ and let rows $r,s \in \rel'$
 	be such that $r[a] \neq s[a]$.
 	The only case where the conclusion $r[X] \neq s[X]$ may possibly be in doubt is 
 	if $r \in \rel'{\setminus}\rel$ and $s \in \rel$
 	(all new rows in $\rel'{\setminus}\rel$ agree on $a$).
	Hence, $r = r_b$ for some $b \neq a$.
	If $b \in X$, the new value ${\times}$ appears in the projection $r[X]$
	but not in $s[X]$;
	otherwise, we have $r[X] = r^*[X]$.
	Since $X \fd a$ holds for the pair $r^*,s  \in \rel$,
	the relation $r[X] = r^*[X] \neq s[X]$ follows.
\end{proof}

The next lemma proves that every instance of 
the unrestricted \FD problem can be expressed
as an equivalent Boolean formula in conjunctive normal form.
Since CNF formulas are exactly the \mbox{$2$-normalized} ones,
we obtain a reduction to \WS{2}.
This is the main result of this section.

\begin{lem}
  \label[lem]{lem:fd_to_ws2}
  There is a parameterized reduction from \emph{\FD} to the \emph{\WS{2}} problem.
\end{lem}

\begin{proof}
	Given a database $\rel$ over $R$, we derive a CNF formula
	that has a satisfying assignment of Hamming weight $k\nwspace{+}\nwspace{1}$
	if and only if there is a non-trivial FD
	with left-hand side of size $k$ that holds in $\rel$.
	We use two types of variables distinguished by their semantic purpose,
	$\var{LHS} = \{x_a \mid a \in R \}$ and 
	$\var{RHS} = \{y_a \mid a \in R \}$.
	Some variable $x_a$ from $\var{LHS}$ being set to \true 
	corresponds to the attribute $a$ appearing on the left-hand side
	of the sought functional dependency;
	for $y_a$ from $\var{RHS}$, this means the attribute is the right-hand side.

	We start with the RHS.
	It might be tempting to enforce
	that any satisfying assignment chooses exactly one variable from $\var{RHS}$.
	We prove below that this is not necessary and forgoing the $\Or(|R|^2)$ clauses
	representing this constraint allows for a (slightly) leaner construction. 
	However, we do have to ensure that at least one variable from $\var{RHS}$ is set to \true
	and the corresponding one in $\var{LHS}$ is \false.
	This is done by the clauses $C_{RHS} = \bigvee_{y_a \in \var{RHS}} y_a$
	and $C_a = \neg y_a \vee \neg x_a$ for each $a \in R$.
	The subformula $\varphi_{RHS}$ is their conjunction.

	We now model the LHS.
	For any attribute $a \in R$ and rows $r,s \in \rel$, let
 	\begin{equation*}
		C_{a,r,s} = \neg y_a \vee \bigvee_{\substack{b \in R{\setminus}\{a\}\\[.125em] r[b] \neq s[b]}} x_b.
	\end{equation*}

	\noindent
	Intuitively, the clause $C_{a,r,s}$ represents the fact
	that if $X \fd a$ is a valid, non-trivial FD, 
	then $X$ has to contain at least one attribute $b$, different from $a$,
	such that $r[b] \neq s[b]$.
  	From these clauses, we assemble the subformula	
	\begin{equation*}
    	\varphi_a = \bigwedge_{\substack{r,s \in \rel\\[.125em] r[a] \neq s[a]}} C_{a,r,s}.
	\end{equation*}
	The output of our reduction is the formula 
	$\varphi = \varphi_{RHS} \wedge \bigwedge_{a \in R} \varphi_a$.
	Indeed, $\varphi$ is in conjunctive normal form
	and has $\Or(|R| \nwspace |\rel|^2)$ clauses with at most $|R|$ literals each.
	An encoding of $\varphi$ is computable in time linear in its size.

	Regarding the correctness of the reduction, recall that we claimed $\varphi$
	to have a weight $k\nwspace{+}\nwspace{1}$ satisfying assignment iff
	a \mbox{non-trivial} functional dependency $X \fd a$ of size $k$ holds in $\rel$.
	Suppose that the latter is true.
	We show that setting the variable $y_a$ as well as all $x_b$ with $b \in X$ to \true
	(and all others to \false) satisfies $\varphi$.
	Note that the assignment automatically satisfies $C_{RHS}$
	and all $C_{b,r,s}$ with $b \neq a$.
	We are left with the subformula $\varphi_a$
	containing the clauses $C_{a, r, s}$ for row pairs with $r[a] \neq s[a]$.
	To distinguish these pairs, the LHS $X$ includes, for each of them,
	some attribute $b \in R{\setminus}\{a\}$ such that $r[b] \neq s[b]$.
	Clause $C_{a, r, s}$ is then satisfied by the corresponding literal $x_b$.
	
	For the other direction, we identify assignments with the variables they set to \true.
	Let $A \subseteq \var{LHS} \cup \var{RHS}$ be an assignment
	of Hamming weight $|A| = k\nwspace{+}\nwspace{1}$ that satisfies $\varphi$.
	The assignment induces two subsets of the schema $R$,
	namely, $X = \{a \in R \mid x_a \in A \cap \var{LHS}\}$ and
	$Y = \{a \in R \mid y_a \in A \cap \var{RHS}\}$.
	Due to the clause $C_{RHS}$, $Y$ is non-empty
	and $X$ contains at most $k$ elements.
	Moreover, $X$ and $Y$ are disjoint as the $C_a$ are all satisfied.
	We say that the generalized functional dependency $X \fd Y$ holds in a database
	if $X \fd a$ holds for every $a \in Y$.
	It is clearly enough to show that $X \fd Y$ indeed holds in $\rel$.
	
	Assume $X \fd a$ fails for some $a \in Y$.
	This is witnessed by a pair of rows $r,s \in \rel$ with $r[X] = s[X]$ but $r[a] \neq s[a]$,
	whence the clause $C_{a,r,s}$ is present in $\varphi$.
	Since $y_a \in A$ is in the assignment, the literal $\neg y_a$ evaluates to \false.
	Also, 	as $X$ is disjoint from the difference set 
	$\{ b \in R{\setminus}\{a\} \mid r[b] \neq s[b] \}$,
	no other literal satisfies $C_{a,r,s}$,
	which is a contradiction.
\end{proof}

We have established a chain of parameterized reductions between the dependency detection problems
of unique column combinations and functional dependencies.
The fact that the endpoints \HS and \WS{2} are both $\W[2]$-complete 
shows that all of the problems are, this completes the proof of \Cref{thm:unique_fd_W[2]}.

\subsection{Approximation and Discovery}
\label{subsec:UCC_discovery}

\noindent
Our reductions have implications beyond the scope of parameterized complexity.
Observe that the transformations in \Cref{lem:hs_to_unique,lem:unique_to_fd}
can be computed in polynomial time and preserve approximations with respect to the
solution size $k$ or $|X|$, respectively.
Also, recall that the minimization version of \HS is \NP-hard to approximate 
within a factor of $(1-\varepsilon) \ln |V|$
for every $\varepsilon > 0$~\cite{Dinur14ParallelRepetition}.
As a consequence, the minimization versions of \Unique, \FDfixed,
and \FD all inherit the same hardness of approximation.
Previously, an approximation-preserving reduction was known
only for the \Unique problem, starting from \textsc{Set Cover}~\cite{Akutsu96ApproxMinKey}.

Rather than approximating minimum solutions, we are mainly interested in the discovery
of minimal dependencies in databases.
Traditionally, enumeration has been studied via embedded decision problems
that are different from those defined in \Cref{subsec:UCC_problem_definitions}.
Instead, the \TransHyp problem (enumerating minimal hitting sets)
has been associated with the problem of deciding for two hypergraphs $\Hyp$ and $\Gyp$
whether $\Gyp  = \Tr(\Hyp)$ or, equivalently, $\Hyp = \Tr(\Gyp)$, called the \textsc{Dual} problem.
\EnumUnique analogously
corresponds to decide for a database $\rel$ and hypergraph $\Hyp$,
whether $\Hyp$ consists of all minimal uniques of $\rel$.
Intuitively, this formalizes the decision whether
an enumeration algorithm has found all solutions.

Both decision problems are in \co\NP and it was proven 
by Eiter and Gottlob~\cite{EiterGottlob95RelatedProblems} that they are many-one equivalent.
Using a lifting result by Bioch and Ibaraki~\cite{Bioch95Identification},
this shows that minimal hitting sets can be enumerated in output-polynomial time
if and only if minimal UCCs can,
which is the case if and only if \textsc{Dual} is in \P.
Such an equivalence is theoretically appealing and has lead to the quasi-output-polynomial
upper bound on the running time of hitting set/UCC enumeration
that is currently the best known~\cite{FredmanKhachiyan96Dualization}.
The connection to enumeration has also inspired the intriguing result that \textsc{Dual} is likely not \co\NP-hard as it can be solved
in polynomial time when given access to $\Or(\sfrac{\log^2 N}{\log\log N})$ suitably guessed
non\-deterministic bits~\cite{Beigel99BoundedNondeterminism,EiterGottlobMakino03NewResults},
where $N$ denotes the total input size of the pair $(\Hyp, \Gyp)$.
There are classes of hypergraphs for which \textsc{Dual} is indeed in \P, see for
example~\cite{Boros98Subimplicants,DomingoMishraPitt99DualizationLearningMembershipQueries,
EiterGottlobMakino03NewResults,Peled94DualRegularBooleanFunction},
or at least in \FPT with respect to certain structural parameters~\cite{ElbassioniHagenRauf08FPTHypergraphDuality}.
For a much more thorough overview of decision problems associated with enumeration,
see the recent work by Creignou et al.~\cite{Creignou19ComplexityTheoryHardEnumerationProblems}.

Unfortunately, the approach described above holds only limited value
when it comes to designing practical algorithms.
Imagine an implementation of the discovery of minimal UCCs of a database $\rel$
via repeated checks whether the hypergraph $\Hyp$ of previously found solutions is already complete.
Such an algorithm is bound to use an amount of memory that is exponential in the size of $\rel$.
This is due to the fact that some databases have exponentially many minimal solutions 
and the decision subroutine at least has to read all of $\Hyp$.
Note that such a large memory consumption is not at all necessary
as there are algorithms known for \TransHyp
whose space complexity is only linear in the input size~\cite{ElbassioniHagenRauf08FPTHypergraphDuality,Murakami14Dualization}.
In fact, enumeration algorithms are often analyzed
not only with respect to their running time, but also in terms of space consumption,
see~\cite{CapelliStrozecki19Incremental,Conte20SublinearSpaceMaximalCliqueEnumeration}.
For data profiling problems like \EnumUnique on the other hand,
space-efficient algorithms
have only recently started to received some attention~\cite{Birnick20HPIValid,Blaesius19EfficientlyALENEX,Koehler16Possible}.

In the following, we simplify and at the same time extend the above equivalences
making them usable in practice,
namely, we prove \Cref{thm:equiv_transHyp}.
It states the existence of parsimonious reductions, in both directions,
that relate the input instances directly on the level of the enumeration problems,
without decision problems as intermediaries.
This way, we characterize the enumeration complexity of unique column combinations
as well as functional dependencies, both with fixed and arbitrary right-hand side.
Curiously, our insights on enumeration also stem from the study of decision problems,
however, the results are entirely lifted.

The reductions between the decision problems
\HS, \Unique, \FDfixed, and \FD (in that order) for \emph{minimum} dependencies
described in \Cref{lem:hs_to_unique,lem:unique_to_fd}
are all built on bijective correspondences between the solutions.
The running times of the reductions are polynomial and independent of the given budget $k$.
Finally, the mappings of the solution spaces also preserve set inclusions.
This means, the same input transformations applied to the discovery of \emph{minimal} dependencies
are in fact parsimonious reductions from \TransHyp to \EnumUnique and 
further to \textsc{Enumerate Minimal FDs/FDs}\textsubscript{\textit{fixed RHS}}.
Regarding the inverse direction, \Cref{obs:UCCs_and_transversals}
shows that the enumeration of UCCs is at most as hard as that of hitting sets,
that is, \EnumUnique and \TransHyp are equivalent.
It is worth noting that the parsimonious reductions increase the size of the instances
at most by a polynomial factor (in most case a constant one) \emph{in the input size only}
and therefore transfer the space complexity of any enumeration algorithm
from one side to the other.

To complete the proof of \Cref{thm:equiv_transHyp},
we still need the following lemma that characterizes the
complexity of functional dependency discovery
in terms of the \TransHypUnion problem.

\begin{lem}
\label[lem]{lem:parsimonious_transHyp_union}
	The problems \emph{\EnumFD} and \emph{\TransHypUnion} are equivalent under parsimonious reductions.
	Moreover, there is a parsimonious reduction from \emph{\EnumFDfixed} to \emph{\TransHyp}.
\end{lem}

\begin{proof}
	This proof uses techniques that already helped to establish the previous lemmas of this section.
	Let $\rel$ be a database over schema $R$, $a \in R$ some attribute,
	$r,s \in \rel$ two rows with $r[a] \neq s[a]$.
	Recall that their difference set is $D(r,s) = \{b \in R \mid r[b] \neq s[b]\}$.
	We define their \emph{punctured difference set} to be $D(r,s){\setminus}\{a\}$.
	It is implicit in the proof of \Cref{lem:fd_to_ws2}\emDash{}and
	easy to verify from the definition of functional dependencies\emDash{}that
	a set $X \subseteq R{\setminus}\{a\}$ is the left-hand side of a
	valid, minimal, non-trivial FD $X \fd a$ if and only if it is a minimal hitting set 
	for the hypergraph $\Dyp_a = \{ D(r,s){\setminus}\{a\} \mid r,s \in \rel; r[a] \,{\neq}\, s[a] \}$
	of punctured difference sets.

	Transforming the input database $\rel$ over schema $R = \{a_1, \dots, a_{|R|}\}$
	into the $|R|$ hypergraphs $\Dyp_{a_1}, \dots, \Dyp_{a_{|R|}}$
	is a parsimonious reduction from the \EnumFD problem to \TransHypUnion.
	In the same fashion, fixing the desired right-hand side in the input reduces 
	\EnumFDfixed to \TransHyp.
	
	\begin{figure}
\centering
	\includegraphics[page=1,scale=.93]{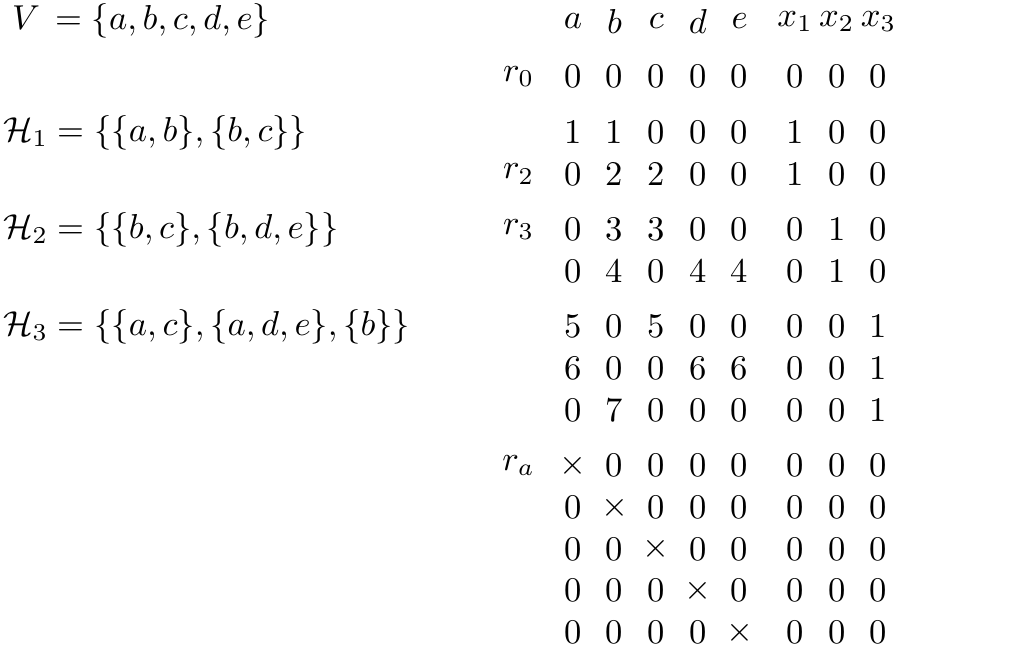}
	\caption{Illustration of \Cref{lem:parsimonious_transHyp_union}.
		The three hypergraphs $\Hyp_1, \Hyp_2, \Hyp_3$ on the vertex set $V$ are on the left
		and the equivalent database $\rel$ is on the right.
		The hypergraphs $\Hyp_1$ and $\Hyp_2$ share the edge $\{b,c\}$,
		but this results in two rows $r_2$ and $r_3$.
		The corresponding transversal hypergraphs are
		$\Tr(\Hyp_1) = \{ \{a,c\}, \{b\}\}$,
		$\Tr(\Hyp_2) = \{ \{b\}, \{c,b\}, \{c,e\}\}$,
		and $\Tr(\Hyp_3) = \{ \{a,b\}, \{b,c,d\}, \{b,c,e\}\}$.
		The functional dependencies $b \fd x_1$ and $b \fd x_2$ indeed hold in $\rel$,
		and adding the attribute $a$ gives $ab \fd x_3$.
		The rows in the last block eliminate all non-trivial functional dependencies
		$X \fd v$ with $v \in V$.}
\label{fig:transHyp_union}
\end{figure}
	
	The opposite reduction is the main part of the lemma.
	We are given hypergraphs $\Hyp_1, \dots, \Hyp_d$,
	without loss of generality all on the same vertex set $V\!$,
	and we need to compute some database $\rel$ such that its valid, non-trivial
	functional dependencies are in one-to-one correspondence with the hitting sets 
	of the $\Hyp_i$.
	An example can be seen in \Cref{fig:transHyp_union}.
	As the relational schema, we take the set $R = V \cup \{x_1, \dots, x_d\}$,
	where the $x_i$ are attributes not previously appearing in $V\!$.
	The construction of $\rel$ starts similarly to \Cref{lem:hs_to_unique}.
	We first add the all-zeros row $r_0$ (with $r_0[a] = 0$ for every $a \in R$).
	Let $m = \sum_{i = 1}^d |\Hyp_i|$ be the total number of edges
	and $E_1, E_2, \dots, E_m$ an arbitrary numbering of them.
	Note that if the same set of vertices is an edge of multiple hypergraphs,
	it appears in the list with that multiplicity.
	For every edge $E_j$, we add the following row $r_j$,
	\begin{equation*}
		r_j[a] =
			\begin{cases}
				j, & \text{if } a \in E_j;\\
				0, & \text{if } a \in V{\setminus}E_j;\\
				1, & \text{if } a = x_i \text{ such that } E_j \in \Hyp_i;\\
				0, & \text{otherwise.}
			\end{cases}
	\end{equation*}
	In other words, the subtuple $r_j[V]$ is the characteristic vector of $E_j$
	only that its non-zero entries are $j$ instead of $1$;
	the subtuple $r_j[\{x_1, \dots, x_d\}]$ has exactly one $1$ at the position 
	corresponding to the hypergraph containing $E_j$.
	The remaining construction of database $\rel$
	uses an idea of \Cref{lem:unique_to_fd}~\ref{case:fdfixed_to_fd}.
	Let $\times$ be a new symbol.
	For every vertex $v \in V$, we add the row $r_v$ 
	with $r_v[v] = \times$ and $r_v[a] = 0$
	for all other attributes $a \in R{\setminus}\{v\}$.
	The database $\rel$ can be obtained in time $\poly(m,|V|)$.
	
	We claim that the minimal, valid, non-trivial functional dependencies of $\rel$
	are exactly those having the form $T \fd x_i$ with $T \in \Tr(\Hyp_i)$.
	The existence of a parsimonious reduction from \TransHypUnion to \EnumFD easily follows from that.
	Let $X \subseteq R$ and \mbox{$a \in R{\setminus}X$} be such that
	the FD $X \fd a$ holds in $\rel$ and is minimal.
	For any $v \in V$, the rows $r_0$ and $r_v$ differ only in attribute $v$,
	therefore $v$ is not the right-hand side of any non-trivial FD,
	whence $a = x_i$ for some $1 \le i \le d$.
	As seen above, 
	the set $X$ must be a minimal transversal of the hypergraph
	$\Dyp_{x_i} = \{ D(r,s){\setminus}\{a\} \mid r,s \in \rel; r[x_i] \,{\neq}\, s[x_i] \}$.
	We are left to prove that $\Dyp_{x_i}$ has the same minimal transversals as $\Hyp_i$.
	Let $r,s \in \rel$ be rows that differ in attribute $x_i$,
	say, $r[x_i] = 1$ and $s[x_i] = 0$.
	We thus have $r = r_j$ for some $1 \le j \le m$.
	The rows $r$ and $s$ share only the value $0$, if any.
	Therefore,
	\begin{equation*}
		D(r,s) =
			\begin{cases}
				E_j \cup \{x_i\}, & \text{ if } s = r_0;\\
				E_j \cup E_k \cup \{x_i, x_\ell\}, & \text{ if } s = r_k 
					\text{ for } 1\,{\le}\,k\,{\le}\,m \text{ such that } E_k \in \Hyp_\ell;\\
				E_j \cup \{x_i,v\}, & \text{ if } s = r_v \text{ for } v \in V\!.
			\end{cases}
	\end{equation*}
	In the second case, note that $\ell \neq i$ since $s[x_i] = 0$.
	The above implies that $\Hyp_i \subseteq \Dyp_{x_i}$;
	moreover, all edges in $\Dyp_{x_i}{\setminus}\Hyp_i$ are supersets of ones in $\Hyp_i$.
	The respective minimizations $\min(\Dyp_{x_i}) = \min(\Hyp_i)$ are thus equal and,
	by duality, also their transversal hypergraphs $\Tr(\Dyp_{x_i}) = \Tr(\Hyp_i)$
	are the same.
\end{proof}
 
It is known that \EnumFD can be solved in output-polynomial time
if and only if \TransHyp can~\cite{EiterGottlob95RelatedProblems},
this has been established along the same lines as discussed in the remarks
preceding~\Cref{lem:parsimonious_transHyp_union}.
Notably, Eiter and Gottlob additionally presented an alternative construction
in the extended version of~\cite{EiterGottlob95RelatedProblems}
that is \emph{almost} parsimonious.
The only condition they needed to relax is the bijection between the solution spaces.
They transform a database over schema $R$ into some hypergraph $(R^2, \Fyp)$
such that the majority of its minimal hitting sets indeed
correspond to the functional dependencies with arbitrary right-hand side.
However, $\Fyp$ has some $\Or(|R|^4)$ excess solutions (that is,
polynomially many in the input size only),
which do not have an FD counterpart, but are easily recognizable.
We leave it as an open problem to give a fully parsimonious reduction between the problems.
We have shown above that this is equivalent to encoding
the hitting set information of $|R|$ different hypergraphs into a single one.

\section{Inclusion Dependencies}
\label{sec:INDs}

\noindent
We now discuss inclusion dependencies in relational databases.
We show that their detection problem, when parameterized by the solution size,
is one of the first natural problems to be complete for the class $\W[3]$.
We do so by proving its FPT-equivalence 
with the weighted satisfiability problem for a certain fragment of propositional logic.
Later in \Cref{subsec:IND_discovery}, we transfer our results 
to the discovery of maximal inclusion dependencies.

\subsection{Problem Definitions}
\label{subsec:IND_problem_definitions}

\noindent
A Boolean formula is \emph{antimonotone} if it only contains negative literals.
We identify a variable assignment with those variables that are set to \true.
In the case of antimonotone formulas, this means that
the satisfying assignments are closed under arbitrarily 
turning variables to \false, that is, taking subsets.
The empty assignment that assigns \false to all variables
is always satisfying.
Recall that a formula is \mbox{$3$-normalized}
if it is a conjunction of disjunctions of conjunctions of literals
or, equivalently, if it is a conjunctions of subformulas in disjunctive normal form (DNF).
An example of an antimonotone, $3$-normalized formula is
\begin{equation*}
	((\neg x_1 \wedge \neg x_2 \wedge \neg x_4) \vee (\neg x_3 \wedge \neg x_4)) 
	\wedge ((\neg x_1 \wedge \neg x_3) \vee (\neg x_2 \wedge \neg x_5) \vee (\neg x_1 \wedge \neg x_4 \wedge \neg x_5)).
\end{equation*}
This formula admits satisfying assignments of Hamming weight~$0$, $1$, and $2$, but none of larger weight.

The \WANS problem is the special case of \WS{3}
restricted to antimonotone formulas.

\vspace*{.6em}
\noindent\WANS (\wans)
\begin{description}
	\item [Instance:] An antimonotone, $3$-normalized Boolean formula $\varphi$
	
		and a non-negative integer $k$.
	\item [Parameter:] The non-negative integer $k$.
	\item [Decision:] Does $\varphi$ admit a satisfying assignment of Hamming weight $k$?
\end{description}

\noindent
By the above remark, this is the same as asking for an assignment of weight at least $k$.
The Antimonotone Collapse Theorem of Downey and 
Fellows~\cite{DowneyFellows95FPTCompletenessI,DowneyFellows95FPTCompletenessII}
implies that the \wans special case is \mbox{$\W[3]$-complete} on its own.

The inclusion-wise maximal satisfying assignments carry the full information about the collection
of all satisfying assignments.
It is therefore natural to define the corresponding enumeration problem as follows.

\vspace*{.6em}
\noindent\EnumWANS
\begin{description}[labelwidth = 2.45cm]
	\item [Instance:] An antimonotone, $3$-normalized Boolean formula $\varphi$.
	\item [Enumeration:] List all maximal satisfying assignments of $\varphi$.
\end{description}

For inclusion dependencies, the situation is similar.
Every subset of a valid IND is also valid.
Asking for a dependency of size exactly $k$ 
is thus the same as asking for one of size at least $k$.
We define two variants of the decision problem,
similar as we did with functional dependencies.
The more restricted variant requires the two databases to have the same schema
with the identity mapping between columns.

\vspace*{.6em}
\noindent\INDfixed
\begin{description}
	\item [Instance:] Two relational databases $\rel$, $\sel$ over schema $R$
	
		and a non-negative integer $k$.
	\item [Parameter:] The non-negative integer $k$.
	\item [Decision:] Is there a set $X \subseteq R$ with $|X| = k$ such that
	
		$\rel[X] \subseteq \sel[X]$ is an inclusion dependency?
\end{description}

\vspace*{.6em}
\noindent\IND
\begin{description}
	\item [Instance:] Two relational databases, $\rel$ over schema $R$ and $\sel$ over $S$,
			
			and a non-negative integer $k$.
	\item [Parameter:] The non-negative integer $k$.
	\item [Decision:] Is there a set $X \subseteq R$ with $|X| = k$ 
			and an injective mapping
			
			$\sigma \colon X \to S$ such that $\rel[X] \subseteq \sel[\sigma(X)]$ is an inclusion dependency?
\end{description}

\noindent
The unparameterized variant of the general \IND problem is \NP-complete
already for pairs of binary databases~\cite{Kantola92FDsAndINDs}.

The solutions of \INDfixed are mere subsets of the underlying schema,
therefore it is clear what we mean by a maximal solution.
The case of the general \IND problem is slightly more intricate.
Recall from \Cref{subsec:prelims_databases}
that we say a general inclusion dependency $(X,\sigma)$ is maximal 
if there is no other IND $(X',\sigma')$ such that $X \subsetneq X'$ is a proper subset
and $\sigma$ is the restriction of $\sigma'$ to $X$.
Note that the pair $(X',\tau)$ with an alternative mapping $\tau$ 
might still be a valid inclusion dependency.
This leads to the following enumeration problems.

\vspace*{.6em}
\noindent\EnumINDfixed
\begin{description}[labelwidth = 2.45cm]
	\item [Instance:] Two relational databases $\rel$, $\sel$ over the same schema.
	\item [Enumeration:] List all maximal valid inclusion dependencies between $\rel$ and $\sel$
	
		\hspace*{0.3cm} with the identity mapping between the columns.
\end{description}

\vspace*{.6em}
\noindent\EnumIND
\begin{description}[labelwidth = 2.45cm]
	\item [Instance:] Two relational databases $\rel$ and $\sel$.
	\item [Enumeration:] List all maximal valid inclusion dependencies between $\rel$ and~$\sel$.
\end{description}

\subsection{Membership in $W[3]$}
\label{subsec:ind-is-in-w3}

\noindent
We show that both variants of the \IND decision problem are contained in the class $\W[3]$.
Recall that the \INDfixed problem 
restricts the input to pairs $(\rel,\sel)$ of databases over the same schema
and forbids solutions in which the set of values $\rel[a]$ of one column 
are contained in $\sel[b]$ for some other column $b \neq a$.
As a first step, we show (not entirely surprisingly) that this variant is at most as hard
as the general problem.

\begin{lem}
\label[lem]{lem:ind-with-or-without-mapping}
	There is a parameterized reduction from \emph{\textsc{Inclusion}}
	\emph{\textsc{Depen\-dency}\textsubscript{\textit{Identity}}} to \emph{\IND}.
\end{lem}

\begin{proof}
	Let $\rel$ and $\sel$ be two databases over the schema $R$ and
	let $t^{\times} = (\times_a)_{a \in R}$ be a new row,
	where the $\times_a$ are $|R|$ different
	symbols none of which are previously used anywhere in $\rel$ or $\sel$.
	In the restricted setting, an inclusion dependency
	is a set $X \subseteq R$ of columns such that $\rel[X] \subseteq \sel[X]$.
	It is easy to see that $(\rel,\sel)$ has such an inclusion dependency of size $k$, for any $k$,
	if and only if $(\rel \cup \{t^{\times}\}, \sel \cup \{t^{\times}\})$
	has an inclusion dependency of the same size with an arbitrary mapping between the columns
	since $\times_a \in \sel[b]$ holds iff $a = b$.
	The lemma follows from here.
\end{proof}

To demonstrate the membership of the general problem in $\W[3]$, we reduce is to \wans.
Namely, we compute from the two databases an antimonotone, $3$-normalized formula
which has a weight $k$ satisfying assignment if and only if the 
databases admit an inclusion dependency of that cardinality.
For this, we use a correspondence between \emph{pairs} of attributes and Boolean variables.

\begin{lem}
\label[lem]{lem:ind-to-wans-without-mapping}
 	There is a parameterized reduction from \emph{\IND} to \emph{\WANS}.
\end{lem}

\begin{proof}
	Let $R = \{a_1, \dots, a_{|R|}\}$ and $S = \{b_1, \dots, b_{|S|}\}$ be two schemas.
	We introduce a Boolean variable $x_{i, j}$ for each pair of attributes $a_i \in R$ and $b_j \in S$.
	We let $\var{P}$ denote the set of variables
	corresponding to a collection $P \subseteq R \times S$ of such pairs.
	Consider a subset $X \subseteq R$ together with an injection $\sigma \colon X \to S$.
	From this, we construct a truth assignment including the variable $x_{i, j}$ (setting it to \true)
	iff $a_i \in X$ and $\sigma(a_i) = b_j$.
	The resulting assignment has weight $|X|$ and the collection of all possible configurations
	$(X,\sigma)$ is uniquely described by $\var{R \times S}$
	and the truth assignments obtained this way.
	Moreover, these assignments all satisfy the following antimonotone Boolean formula $\varphi_{map}$.
	\begin{align*}
		\varphi_{map} = 
			\left(\bigwedge_{i = 1}^{|R|} \bigwedge_{j = 1}^{|S| - 1} \bigwedge_{j' = j+1}^{|S|}
				(\neg x_{i, j} \vee \neg x_{i, j'})\right)
			\  \wedge \ 
			\left(\bigwedge_{j = 1}^{|S|} \bigwedge_{i = 1}^{|R| - 1} \bigwedge_{i' = i+1}^{|R|}
				(\neg x_{i, j} \vee \neg x_{i', j})\right).
  \end{align*}

	\noindent
	The first half of $\varphi_{map}$ expresses that,
	for every pair of variables $x_{i, j}$ and $x_{i, j'}$ with $j \neq j'$,
	at most one of them shall be \true;
	the second half is satisfied if the same holds for all pairs $x_{i, j}$ and $x_{i', j}$
	with $i \neq i'$.
	Conversely, a satisfying assignment $A$ (a subset of $\var{R \times S}$) for $\varphi_{map}$
	defines a relation $\sigma \subseteq R \times S$ and a set $X \subseteq R$ 
	by setting $\sigma = \{ (a_i, b_j)  \mid x_{i, j} \in A \}$
	and $X = \{a_i \in R \mid \exists\, 1 \le j \le |S| \colon x_{i, j} \in A\}$.
	By construction, the relation $\sigma$ is not only a function $\sigma \colon X \to S$,
	but an injection.
	In summary, $\varphi_{map}$ is fulfilled exactly by the assignments described above.
	Observe that $\varphi_{map}$ is in conjunctive normal form 
	and therefore also $3$-normalized
	as each literal is a conjunctive clause of length 1.

	We now formalize the requirement that a configuration $(X, \sigma)$
	is an inclusion dependency in a given pair of databases $\rel$ and $\sel$
	over the respective schemas $R$ and $S$,
	that is, that $\rel[X] \subseteq \sel[\sigma(X)]$ holds.
	First, assume that each database consists only of a single row
	$r_\ell$ and $s_m$, respectively.
	We say a pair of attributes $(a_i, b_j) \in R \times S$ is \emph{forbidden}
	for $r_\ell$ and $s_m$ if $r_\ell[a_i] \neq s_m[b_j]$.
 	Let $F_{\ell, m}$ be the set of all forbidden pairs.
 	For an configuration $(X, \sigma)$ to be an IND,
 	the variables $x_{i,j}$ need to be set to \false for all $(a_i, b_j) \in F_{\ell, m}$.
 	In terms of Boolean formulas, this is represented by the conjunctive clause
	$M_{\ell, m} = \bigwedge_{x \in \var{F_{\ell, m}}}\! \neg x$.
	It follows that $(X, \sigma)$ is an inclusion dependency if and only if 
	the corresponding variable assignment satisfies both $\varphi_{map}$ and $M_{\ell, m}$.

	Now suppose $\sel$ has multiple rows, while $\rel$ is still considered to have only one.
	The configuration $(X, \sigma)$ is an IND for $(\rel, \sel)$
	iff it is one for at least one instance $(\rel, \{s_m\})$ with $s_m \in \sel$.
	If also $\rel$ has more records, then $(X, \sigma)$ is an IND for $(\rel,\sel)$ 
	iff it is one in each instance $(\{r_\ell\}, \sel)$ with $r_\ell \in \rel$.
	Therefore, we obtain an inclusion dependency if and only if
	$\varphi_{map}$ and the formula
    \begin{align*}
    	\varphi = \bigwedge_{r_{\ell} \in \rel} \bigvee_{s_m \in \sel} M_{\ell, m}
	\end{align*}
	are simultaneously satisfied by the assignment corresponding to $(X,\sigma)$.
  
	The formula $\varphi \wedge \varphi_{map}$ is antimonotone and 3-normalized.
	The (disjunctive) clauses of $\varphi_{map}$ can be constructed
	in total time $\Or(|R|^2 |S| + |R||S|^2)$
	and all sets $F_{\ell,m}$ together are computable in time $\Or(|\rel||\sel||R||S|)$.
	An encoding of $\varphi \wedge \varphi_{map}$ can thus be obtained
	from the input databases $\rel$ and $\sel$ in polynomial time.
	Finally, by the above observation that any solution for the
	sub-formula $\varphi_{map}$ that corresponds to $(X,\sigma)$
	has weight $|X|$, the reduction preserves the parameter.
\end{proof}

\subsection{Hardness for $W[3]$}
\label{subsec:ind-w3-hard}

\noindent
We now show that detecting inclusion dependencies is also hard for $\W[3]$.
We argue that the existence of weighted satisfying assignments
for \mbox{3-normalized}, antimonotone formulas can be decided by solving instances of 
the more restricted \INDfixed variant.
For the reduction, we make use of indicator functions.
On the one hand, we interpret propositional formulas $\varphi$ over $n$ variables 
as Boolean functions $f_\varphi \colon \{0, 1\}^n \to \{0, 1\}$ in the obvious way.
On the other hand, for a pair of databases $\rel$ and $\sel$ over the the same schema $R$,
we represent any subset $X \subseteq R$ by its characteristic vector of length $|R|$.
We then define the \emph{indicator function}
$f_{(\rel, \sel)} \colon \{0, 1\}^{|R|} \to \{0, 1\}$ by requiring that
$f_{(\rel, \sel)}(X) = 1$ holds if and only if $X$ is an inclusion dependency
(with the identity mapping between the columns).

We claim that for any formula $\varphi$ that is antimonotone and 3-normalized,
there is a pair $(\rel, \sel)$ of databases computable in polynomial time  
such that $f_\varphi = f_{(\rel, \sel)}$.
Clearly, this gives a parameterized reduction from \wans to the \INDfixed problem.
The remainder of this section is dedicated to proving this claim.
Recall that the top level connective of a \mbox{3-normalized} formula is a conjunction.
We start by demonstrating how to model this using databases.

\begin{lem}
\label[lem]{lem:ind-hardness_conjunction}
	Let $(\rel^{(1)}, \sel^{(1)})$ and $(\rel^{(2)}, \sel^{(2)})$ be two
	pairs of databases, all over the same schema $R$,
	with indicator functions $f^{(1)}$ and $f^{(2)}$, respectively.
	There exists a polynomial time computable pair $(\rel, \sel)$ over $R$
	of size $|\rel| = |\rel^{(1)}| + |\rel^{(2)}|$ and $|\sel| = |\sel^{(1)}| + |\sel^{(2)}|$,
	having indicator function $f_{(\rel, \sel)} = f^{(1)} \wedge f^{(2)}$.
\end{lem}

\begin{proof}
	Without loosing generality, the values appearing in $\rel^{(1)}$ and $\sel^{(1)}$
	are disjoint from those in $\rel^{(2)}$ and $\sel^{(2)}$.
	We straightforwardly construct $(\rel, \sel)$ 
	as $\rel = \rel^{(1)} \cup \rel^{(2)}$ and $\sel = \sel^{(1)} \cup \sel^{(2)}$,
	which matches the requirements on both the computability and size.
	We still need to show $f_{(\rel, \sel)} = f^{(1)} \wedge f^{(2)}$.

	Equivalently, we prove that a set $X \subseteq R$ is an inclusion dependency in $(\rel, \sel)$
	if and only if it is one in both pairs
	$(\rel^{(1)}, \sel^{(1)})$ and $(\rel^{(2)}, \sel^{(2)})$. 
	Let $X$ be an IND in $(\rel^{(1)}, \sel^{(1)})$ as well as $(\rel^{(2)}, \sel^{(2)})$.
	That means, for every row $r \in \rel^{(1)}$,
	there exists some $s \in \sel^{(1)}$ with $r[X] = s[X]$; 
	same for $\rel^{(2)}$ and $\sel^{(2)}$.
	As all those rows are also present in $(\rel, \sel)$, $X$ is an IND there as well.
	Conversely, suppose
	$X$ is not an inclusion dependency in, say, $(\rel^{(1)},\sel^{(1)})$.
	Then, $\rel^{(1)}$ has a row $r$ that disagrees with every $s \in \sel^{(1)}$
	on some attribute in $X$.
	The record $r$ is also in present in $\rel$
	and all rows in $\sel$ belong either to $\sel^{(1)}$ or have completely disjoint values.
	This results in $r[X] \not= s[X]$ for \emph{every} record $s \in \sel$.
\end{proof}

One could hope that there is a similar method treating disjunctions.
However, we believe that there is none that is both computable in \mbox{FPT-time} 
and compatible with a complementing method representing conjunctions
(for example, the one above).
The reason is as follows.
Negative literals are easily expressible by pairs of single-row databases.
Together with \mbox{FPT-time} procedures of constructing conjunctions as well as disjunctions,
one could encode antimonotone Boolean formulas of arbitrary logical depth.
The Antimonotone Collaps Theorem~\cite{DowneyFellows13Parameterized},
states that the \textsc{Weighted Antimonotone $t$-normalized Satisfiability} problem
is $\W[t]$-complete for every odd $t \ge 3$.
This would then render \INDfixed to be hard for all classes $\W[t]$ and,
as a consequence of \Cref{lem:ind-to-wans-without-mapping,lem:ind-with-or-without-mapping},
the $\W$-hierarchy would collapse to its third level.
That being said, there is a method specifically tailored to antimonotone DNF formulas.

\begin{lem}
\label[lem]{lem:ind-hardness_DNF}
	Let $\varphi$ be an antimonotone formula in disjunctive normal form.
	There are relational databases $\rel$ and $\sel$ over the same schema
	computable in time polynomial in the size of $\varphi$
	such that $f_\varphi = f_{(\rel, \sel)}$.
\end{lem}

\begin{proof}
	Let $x_1, \dots, x_n$ be the variables of $\varphi$.
	Define schema $R = \{a_1, \dots, a_n\}$ by identifying variable $x_i$ with attribute $a_i$.
	We first describe how to obtain $\rel$ and subsequently construct a matching database $\sel$.
	Let $M_1, \dots, M_m$ denote the $m$ constituting conjunctive clauses of the DNF formula $\varphi$.
	For each $M_j$, we define the row $r_j$, similarly as in the proof of \Cref{lem:hs_to_unique}, as
	\begin{equation*}
		r_j[a_i] =
			\begin{cases}
				j, & \text{if variable } x_i \text{ occurs in } M_j;\\
				0, & \text{otherwise.}
			\end{cases}
	\end{equation*}
	See \Cref{fig:ind-hardness} for an example.
	
	\begin{figure}
	\centering
	\includegraphics[page=1,scale=1.1]{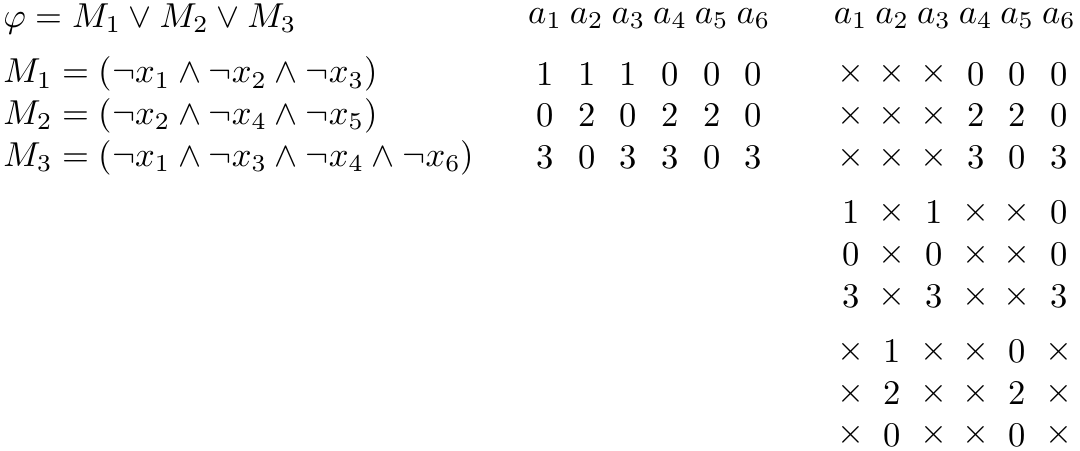}
	\caption{Illustration of \Cref{lem:ind-hardness_DNF}.
		The antimonotone DNF formula $\varphi$ on the left 
		has the three conjunctive clauses $M_1, M_2, M_3$.
		The equivalent instance of \INDfixed consists of database
		$\rel$ in the center and $\sel$ on the right.
		There are three maximal inclusion dependencies
		$\{a_4, a_5, a_6\}$, $\{a_1, a_3, a_6\}$, and $\{a_3, a_5\}$.
		Adding any more attributes to either of them
		would create a hitting set for the conjunctive clauses,
		corresponding to an unsatisfying assignment.
	}
\label{fig:ind-hardness}
\end{figure}
	
	The second database $\sel$ is constructed by first creating $m$ copies of $\rel$.
	Let $\times$ be a new symbol not appearing anywhere in $\rel$.
	In the \mbox{$j$-th} copy of $\rel$, we set the value for attribute $a_i$ to $\times$
	whenever $x_i$ occurs in $M_j$.
	(See \Cref{fig:ind-hardness} again.)
	Note that $|R| = n$ equals the number of variables of $\varphi$
	and $|\rel|$ is linear in the number $m$ of conjunctive clauses, while $|\sel|$ is quadratic.
	The time to compute the pair $(\rel, \sel)$ is linear in their combined size
	and polynomial in the size of $\varphi$.
	It is left to show that the indicator function satisfies $f_\varphi = f_{(\rel, \sel)}$.

	First, suppose $f_\varphi(X) = 1$ for some length-$n$ binary vector $X$
	or, equivalently, for some subset $X \subseteq R$.
	We show that $f_{(\rel, \sel)}(X) = 1$,
	meaning that $X$ is an inclusion dependency in $(\rel,\sel)$.
	Necessarily, we have $f_{M_j}(X) = 1$ for at least one conjunctive clause $M_j$.
	Since $M_j$ contains only negative literals, all of its variables evaluate to \false.
	This is equivalent to $X$ not containing any attribute that corresponds to a variable in $M_j$.  
	In the \mbox{$j$-th} copy of $\rel$ in the database $\sel$,
	the values were changed to $\times$ for exactly those attributes.
	Thus, the projection $\sel[X]$ contains an exact copy of $\rel[X]$
	and $X$ is indeed an IND,
	resulting in $f_{(\rel, \sel)}(X) = 1$.

	For the opposite direction, suppose $f_\varphi(X) = 0$.
	Each conjunctive clause thus contains a variable corresponding to some attribute in $X$.
	Consequently, in each row of $\sel$,
	there is an attribute in $X$ whose value was replaced by $\times$.
	As $\rel$ does not contain the symbol $\times$ at all, $X$ is not an IND
	and $f_{(\rel, \sel)}(X) = 0$.
\end{proof}

\Cref{lem:ind-hardness_conjunction,lem:ind-hardness_DNF} imply that,
given an antimonotone, \mbox{3-normalized} formula $\varphi$, 
we can build an instance $(\rel, \sel)$ of \INDfixed in \mbox{FPT-time} (even polynomial)
such that $f_\varphi = f_{(\rel, \sel)}$.
Together with the findings of \Cref{subsec:ind-is-in-w3},
this finishes the proof of \Cref{thm:ind_W[3]}.

\subsection{Discovery}
\label{subsec:IND_discovery}

\noindent
As we did with minimal unique column combinations and functional dependencies,
we can lift our results from detecting a single inclusion dependency to discovering all of them. 
It turns out that there is a parsimonious equivalence with the enumeration of
assignments to antimonotone, \mbox{3-normalized} formulas, as detailed in \Cref{thm:enum_ind}.
The key observations to prove this are once again that the reductions above
are polynomial time computable,
independently of the parameter, and that they preserve inclusions.

\Cref{lem:ind-hardness_conjunction,lem:ind-hardness_DNF}
describe how to turn the formula $\varphi$ in polynomial time into a pair 
of databases over the common schema $R$,
which is effectively the same as $\var{\varphi}$,
such that the inclusion dependencies $X \subseteq R$ are in canonical
correspondence with the satisfying assignments $A \subseteq \var{\varphi}$.
Moreover, the parameterized reduction preserves inclusion relations between the solutions
such that the maximal dependencies also correspond to the maximal assignments.
In other words, \Cref{lem:ind-hardness_conjunction,lem:ind-hardness_DNF}
induce a parsimonious reduction from \EnumWANS to \EnumINDfixed.
It is also easy to see that \Cref{lem:ind-with-or-without-mapping}
implies a parsimonious reduction from the enumeration
of such restricted inclusion dependencies to the general \EnumIND problem.
The lemma does nothing else but invalidating all non-identity
mappings between the columns.
Finally, \Cref{lem:ind-to-wans-without-mapping} shows how to
translate general inclusion dependencies back to antimonotone, 3-normalized formula $\varphi$.
Observe that the (inclusion-wise) maximal satisfying assignments of the resulting formula
correspond exactly to the notion of maximality for general INDs 
(see \Cref{subsec:prelims_databases,subsec:IND_problem_definitions} for details).
This shows the equivalence of all the enumeration problems involved.
Again, the space complexity is preserved up to polynomial factors by the parsimonious reductions.

We complete the proof of \Cref{thm:enum_ind} by showing that the problems
are at least as hard as \TransHyp.
This is an easy exercise using the structure of antimonotone CNFs.

\begin{lem}
\label[lem]{lem:enum_ind_transversal_hard}
	The enumeration of maximal satisfying assignments of antimonotone Boolean formulas 
	in conjunctive normal form
	is equivalent to \emph{\TransHyp} under parsimonious reductions.
	In particular, \emph{\EnumWANS} is at least as hard as the \emph{\TransHyp} problem.
\end{lem}

\begin{proof}
	For the reductions, in both directions,
	we identify the (disjunctive) clauses of an antimonotone CNF formula $\varphi$
	with the sets of variables they contain.	
	To spell it out, let $\var{\varphi}$ be the set of all variables of $\varphi$
 	and $C_1, \dots, C_m \subseteq \var{\varphi}$ the clauses.
 	Since the the formula is antimonotone,
 	any $C_i$ is satisfied iff there is a variable $x \in C_i$
 	that is assigned \false.
	In other words, an assignment $A \subseteq \var{\varphi}$ (the set of \true variables)
	is satisfying iff its complement $\overline{A} = \var{\varphi}{\setminus}\nwspace{A}$
	is a hitting set of the hypergraph $\{C_i\}_{i \in [m]}$.
	Assignment $A$ is maximal in that regard iff $\overline{A} \in \Tr(\{C_i\}_{i})$ is a minimal transversal.
	In the very same fashion, we can construct from any hypergraph $(\Hyp,V)$
	an antimonotone CNF formula on the variable set $\{x_v \mid v \in V\}$
	by setting $\varphi = \bigwedge_{E \in \Hyp} \bigvee_{v \in E} \neg x_v$.
	Complementing the maximal satisfying assignments of $\varphi$
	recovers the minimal hitting sets of $\Hyp$.
	
	The second part of the lemma follows from any CNF formula being also 3-normalized
	by viewing literals as conjunctive clauses of length 1.
\end{proof}

\section{Conclusion}
\label{sec:conclusion}

\noindent
We have determined the complexity of the detection problems 
for various types of multi-column dependencies,
parameterized by the solution size.
Our results imply that these problems do not admit FPT-algorithms
unless the $\W$-hierarchy at least partially collapses. 
In fact, the detection of inclusion dependencies turned out to be surprisingly hard
in that it is $\W[3]$-complete.
Therefore, a small solution size alone is not enough to explain the good performance in practice.
This is unfortunate as the choice of parameter appears to be very natural in the sense that the requirement of a small solution size is regularly met in practice. 
Of course, our results do not preclude FPT-algorithms for \emph{other} parameters.
As an example, \Unique on databases over schema $R$ 
are trivially in \FPT with respect to the parameter $|R|$ by checking all subsets.
This is of course not very satisfying since assuming the schema to be small is much stronger
than assuming the solutions to be small.
Similarly, one could consider the maximum number $d$ of attributes on which
two rows in the data disagree as parameter.
Using the standard bounded search tree of height $k$ with nodes of degree at most $d$
gives an FPT-algorithm with respect to the parameter $d + k$.
Again, the assumption that any two rows in a relation differ only on a few columns 
seems to be fairly unrealistic for most practical data sets.
More structural research into relational databases is needed to bridge the gap between
the worst-case hardness and empirical tractability of dependency detection problems.
This will involve identifying properties of realistic instances that can explain and hopefully
even improve the running times of practical methods.
For example, by designing multivariate algorithms with more than one parameter.

On the other hand, our results regarding the discovery of all dependencies 
of a certain type in a database
are indeed able to explain the good run times in practice.
We proved that the profiling of relational data at its core is
closely related to the transversal hypergraph problem.
Although the exact complexity of the latter is still open,
there are many empirically efficient algorithms known for it.
Most importantly, modern algorithms for the enumeration of hitting sets
have the advantage that their space complexity is only linear
in the input size, this is a feature many data profiling algorithms are still lacking today~\cite{Abedjan18DataProfiling,Papenbrock16HyFD,Wei19DiscoveryEmbeddedUnique}.
Even more than large run times, prohibitive memory consumption still puts certain databases out of reach
for data profiling today.
The relations shown here may therefore be a way to improve the current state of the art
algorithms in that direction.

\section*{Ackowledgements}

\noindent
We thank Sebastian Kruse, Felix Naumann, and Thorsten Papenbrock
for proposing data profiling as a research topic to us,
for their continued support, and for providing the data presented in the introduction.
We are also thankful for the discussion with Jan Kossmann on query optimization
and with Yann Strozecki on enumeration complexity,
both provided valuable insights and pointers to the literature.

\bibliographystyle{plainurl} 
\bibliography{references}

\end{document}